\newif\ifabstract

\abstracttrue
\abstractfalse 
\newif\iffull
\ifabstract \fullfalse \else \fulltrue \fi

\ifabstract
\documentclass{sig-alternate}
\fi

\iffull
\documentclass[11pt]{article}
\advance \oddsidemargin by -0.8in
\newcommand{\myparskip}{3pt}
\parskip \myparskip
\fi

\usepackage{amsfonts}
\usepackage{amssymb}
\usepackage{amstext}
\usepackage{amsmath}
\usepackage{xspace}
\usepackage{theorem}
\usepackage{graphicx}
\usepackage{graphics}
\usepackage{colordvi}
\usepackage{colordvi}
\usepackage{url,hyperref}

\iffull
        {\hspace*{\fill}$\Box$\par\vspace{4mm}}
\fi

\newcommand{\GM}{Grid-Minor\xspace}
\newcommand{\FPT}{Fixed-Parameter Tractability\xspace}

\newcommand{\algsc}{\ensuremath{{\mathcal{A}}_{\mbox{\textup{\tiny{\sc ARV}}}}}\xspace}
\newcommand{\alphasc}{\ensuremath{\alpha_{\mbox{\tiny{\sc ARV}}}}}
\newcommand{\alphaSC}{\ensuremath{\alpha_{\mbox{\tiny{\sc ARV}}}}}
\newcommand{\alphaARV}{\ensuremath{\alpha_{\mbox{\tiny{\sc ARV}}}}}

\newcommand{\gammaARV}{\ensuremath{\gamma_{\mbox{\tiny{\sc ARV}}}}}
\newcommand{\gammasc}{\ensuremath{\gamma_{\mbox{\tiny{\sc ARV}}}}}

\newcommand{\algSC}{\algsc}

\newcommand{\G}{{\mathbf{G}}}
\renewcommand{\H}{{\mathbf{H}}}

\newcommand{\ceil}[1]{\ensuremath{\left\lceil#1\right\rceil}}
\newcommand{\floor}[1]{\ensuremath{\left\lfloor#1\right\rfloor}}

\newcommand{\event}{{\cal{E}}}

\newcommand{\polylog}[1]{\mathrm{polylog(#1)}}

\newcommand{\opt}{\mathsf{OPT}}

\newcommand{\alphawl}{\ensuremath{\alpha_G}}
\newcommand{\alphaWL}{\alphawl}

\newcommand{\set}[1]{\left\{ #1 \right\}}
\newcommand{\sse}{\subseteq}

\newcommand{\tset}{{\mathcal T}}

\newcommand{\pset}{{\mathcal{P}}}

\newcommand{\bset}{{\mathcal{B}}}
\newcommand{\aset}{{\mathcal{A}}}

\newcommand{\cset}{{\mathcal{C}}}
\newcommand{\fset}{{\mathcal{F}}}

\newcommand{\xset}{{\mathcal{X}}}
\newcommand{\wset}{{\mathcal{W}}}

\newcommand{\rset}{{\mathcal{R}}}

\newcommand{\sset}{{\mathcal{S}}}

\newcommand{\nots}{\overline S}

\newcommand{\be}{\begin{enumerate}}
\newcommand{\ee}{\end{enumerate}}
\newcommand{\bd}{\begin{description}}
\newcommand{\ed}{\end{description}}
\newcommand{\bi}{\begin{itemize}}
\newcommand{\ei}{\end{itemize}}

\newtheorem{lemma}{Lemma}[section]
\newtheorem{theorem}{Theorem}[section]
\newtheorem{observation}{Observation}[section]
\newtheorem{corollary}{Corollary}[section]
\newtheorem{claim}{Claim}[section]
\newtheorem{proposition}{Proposition}[section]

\newtheorem{definition}{Definition}[section]
\iffull
\newenvironment{proof}{\par \smallskip{\bf Proof:}}{\hfill\stopproof}
\def\stopproof{\square}
\def\square{\vbox{\hrule height.2pt\hbox{\vrule width.2pt height5pt \kern5pt
\vrule width.2pt} \hrule height.2pt}}
\fi

\renewcommand{\phi}{\varphi}

\newcommand{\half}{\ensuremath{\frac{1}{2}}}

\newcommand{\poly}{\operatorname{poly}}

\newcommand{\reals}{{\mathbb R}}

\newcommand{\prob}[2][]{\text{\bf Pr}_{#1}\left [#2\right]}

\newcommand{\mynote}[1]{{\sc\bf{[#1]}}}

\newcommand{\dmax}{d_{\mbox{\textup{\footnotesize{max}}}}}

\newcommand{\out}{\operatorname{out}}

\newcommand{\Erdos}{Erdos\xspace}
\newcommand{\Posa}{P\'{o}sa\xspace}
\newcommand{\EP}{{\Erdos-\Posa}\xspace}

\newcommand{\headers}[3]{
\newpage\setcounter{page}{1}
\def\@oddhead{$\underline{\hbox to\textwidth{%
\textbf{\rlap{#1}\phantom{hj}\hfill #2 \hfill \llap{#3}}}}$}
\def\@oddfoot{\hfill\thepage\hfill}}

\newtheorem{prop}[lemma]{Proposition}
\newtheorem{remark}[lemma]{Remark}
\newtheorem{conjecture}{Conjecture}

\def\etal{et al.\xspace}

\def\floor#1{\lfloor {#1} \rfloor}
\def\ceil#1{\lceil {#1} \rceil}
\def\script#1{\mathcal{#1}}
\def\card#1{|#1|}
\def\set#1{\{#1\}}

\def\polylog{\mathrm{poly}\log}
\def\out{\mathrm{out}}
\def\bd{\mathrm{bd}}

\def\mF{\script{F}}
\def\mP{\script{P}}

\def\mX{\script{X}}

\def\cKRV{\gamma_{\mathrm{CMG}}}

\newcommand{\tw}{\mathrm{tw}}
\newcommand{\BN}{\mathrm{BN}}

\begin{document}
\title{Large-Treewidth Graph Decompositions and Applications\footnote{An extended abstract of this paper is to appear in {\em Proc.\ of ACM STOC}, 2013.}}
\author{Chandra Chekuri\thanks{Dept.\ of Computer Science, University
of Illinois, Urbana, IL 61801. {\tt chekuri@illinois.edu}.
Supported in part by NSF grant CCF-1016684 and 
by TTI Chicago on a sabbatical visit.}
\and 
Julia Chuzhoy \thanks{Toyota Technological Institute, Chicago, IL
60637. Email: {\tt cjulia@ttic.edu}. Supported in part by NSF CAREER 
grant CCF-0844872 and Sloan Research Fellowship.}
}

\date{\today}
\iffull
\begin{titlepage}
\maketitle
\thispagestyle{empty}
\fi

\begin{abstract}
  Treewidth is a graph parameter that plays a fundamental role in
  several structural and algorithmic results. We study the problem of
  decomposing a given graph $G$ into node-disjoint subgraphs, where
  each subgraph has sufficiently large treewidth. We prove two
  theorems on the tradeoff between the number of the desired subgraphs
  $h$, and the desired lower bound $r$ on the treewidth of each
  subgraph. The theorems assert that, given a graph $G$ with treewidth
  $k$, a decomposition with parameters $h,r$ is feasible whenever
  $hr^2 \le k/\polylog(k)$, or $h^3r \le k/\polylog(k)$ holds.  We
  then show a framework for using these theorems to bypass the
  well-known \GM Theorem of Robertson and Seymour in some
  applications. In particular, this leads to substantially improved
  parameters in some \EP-type results, and faster algorithms
  for a class of fixed-parameter tractable problems.
\end{abstract}

\iffull
\end{titlepage}
\fi

\section{Introduction}
\label{sec:intro}

Let $G=(V,E)$ be an undirected graph. We assume that the reader is
familiar with the notion of treewidth of a graph $G$, denoted by
$\tw(G)$. 
The main question considered in this paper is the following.  Given an
undirected graph $G$, and integer parameters $h, r < \tw(G)$, can $G$
be partitioned into $h$ node-disjoint subgraphs $G_1,\ldots,G_h$ such
that for each $i$, $\tw(G_i) \ge r$? It is easy to see that for this
to be possible, $h r \le \tw(G)$ must hold. Moreover, it is not hard
to show examples of graphs $G$, where even for $r=2$, the largest
number of node-disjoint subgraphs of $G$ with treewidth at least $r=2$
is bounded by $h = O\left (\frac{\tw(G)}{\log
    (\tw(G))}\right)$.\footnote{Consider a constant degree $n$-node
  expander $G$ with girth $\Omega(\log n)$; the existence of such
  graphs can be shown by the probabilistic method. Let
  $G_1,\ldots,G_h$ be any collection of node-disjoint subgraphs of $G$
  of treewidth at least $2$ each. Then each graph $G_i$ must contain a
  cycle, and by the lower bound on the girth of $G$, $|V(G_i)| =
  \Omega(\log n)$, implying that $h = O(n/\log n)$. On the other hand
  $\tw(G) = \Omega(n)$.}  In this paper we prove the following two
theorems, that provide sufficient conditions for the existence
of a decomposition with parameters $h,r$.

\begin{theorem}
  \label{thm:main1-intro}
  Let $G$ be any graph with $\tw(G) = k$, and let $h,r$ be any
  integers with $h r^2 \le k/\poly\log k$. Then there is an efficient\footnote{In this paper we use the term efficient algorithm to refer to a randomized algorithm
    that runs in time polynomial in $|V(G)|$ and $k$.}
  algorithm to partition $G$ into $h$ node-disjoint subgraphs
  $G_1,\ldots,G_h$ such that $\tw(G_i) \ge r$ for each $i$.
\end{theorem}

\begin{theorem}
  \label{thm:main2-intro}
  Let $G$ be any graph with $\tw(G) = k$, and let $h,r$ be any
  integers with $h^3 r \le k/\poly\log k$. Then there is an efficient
  algorithm to partition $G$ into $h$ node-disjoint subgraphs
  $G_1,\ldots,G_h$ such that $\tw(G_i) \ge r$ for each $i$.
\end{theorem}

We observe that the two theorems give different tradeoffs, depending on
whether $r$ is small or large. It is particularly useful in
applications that the dependence is linear in one of the
parameters. We make the following conjecture, that would
strengthen and unify the preceding theorems.

\begin{conjecture}
  Let $G$ be any graph with $\tw(G) = k$, and let $h,r$ be any
  integers with $h r \le k/\poly\log k$. Then $G$ can be partitioned
  into $h$ node-disjoint subgraphs $G_1,\ldots,G_h$ such that
  $\tw(G_i) \ge r$ for each $i$.
\end{conjecture}

\smallskip
\noindent {\bf Motivation and applications.} The starting point for
this work is the observation that a special case of
Theorem~\ref{thm:main2-intro}, with $h = \Omega(\log^2 k)$, is a
critical ingredient in recent work on poly-logarithmic approximation
algorithms for routing in undirected graphs with constant congestion
\cite{Chuzhoy11,ChuzhoyL12,ChekuriE13}. In particular,
\cite{Chuzhoy11} developed such a decomposition for edge-disjoint
routing, and subsequently \cite{ChekuriE13} extended it to the
node-disjoint case. However, in this paper, we are motivated by a
different set of applications, for which Theorem~\ref{thm:main1-intro}
is more suitable. These applications rely on the seminal work of
Robertson and Seymour \cite{RS-grid}, who showed that there is a large
grid minor in any graph with sufficiently large treewidth.  The
theorem below, due to Robertson, Seymour and Thomas~\cite{RobertsonST94}, gives an improved
quantitative bound relating the size of the grid minor and the treewidth.

\begin{theorem}[Grid-Minor Theorem \cite{RobertsonST94}]
  \label{thm:RST-grid} Let $G$ be any graph, and $g$ any integer, such
  that $\tw(G)\ge 20^{2g^5}$. Then $G$ contains a $g \times g$ grid as
  a minor.  Moreover, if $G$ is planar, then $\tw(G) \ge 6g-4$ suffices.
\end{theorem}

Kawarabayashi and Kobayashi~\cite{KawarabayashiK-grid} obtained an
improved bound of $2^{O(g^2 \log g)}$ on the treewidth required to
ensure a $g\times g$ grid minor, and a further improvement to a bound
of $2^{O(g \log g)}$ was recently claimed by Seymour
\cite{Seymour-grid}.

Notice that Theorem~\ref{thm:RST-grid} guarantees a grid minor of size
sub-logarithmic in the treewidth $k$ in general graphs, and of size
$\Omega(k)$ in planar graphs. Demaine and Hajiaghayi
\cite{DemaineH-grid} extended the linear relationship between the grid
minor size and the treewidth to graphs that exclude a fixed graph $H$
as a minor (the constant depends on the size of $H$, see
\cite{KawarabayashiK-grid} for an explicit dependence). A $g \times g$
grid has treewidth $g$, and it can be partitioned into $h$
node-disjoint grids of size $r \times r$ each, as long as $r\sqrt h =
O(g)$. Thus, in a general graph $G$ of treewidth $k$, the \GM Theorem
currently only guarantees that for any integers $h,r$ with $hr^2 =
O(\log^{2/5} k)$, there is a partition of $G$ into $h$ node-disjoint
subgraphs of treewidth at least $r$ each. Robertson \etal \cite{RobertsonST94}
observed that, in order for $G$ to
contain a $g \times g$ grid as a minor, its treewidth may need to be
as large as $\Omega(g^2 \log g)$, and they suggest that this may
be sufficient. Demaine \etal \cite{DemaineHK09} conjecture that
the treewidth of $\Theta(g^3)$ is both necessary and sufficient.

The existence of a polynomial relationship between the grid-minor size
and the graph treewidth is a fundamental open question, that appears
to be technically very challenging to resolve.  Our work is motivated
by the observation that the \GM Theorem can be bypassed in various
applications by using Theorems~\ref{thm:main1-intro} and
\ref{thm:main2-intro}.  We describe two general classes of such
applications below.

\smallskip
{\em Bounds for \EP type results.} 
The duality between packing and covering plays a central role in graph
theory and combinatorial optimization. One central result of this
nature is Menger's theorem, which asserts that for any graph $G$,
subsets $S$, $T$ of its vertices, and an integer $k$, either $G$
contains $k$ node-disjoint paths connecting the vertices of $S$ to the
vertices of $T$, or there is a set $X$ of at most $k-1$ vertices,
whose removal disconnects all such paths. \Erdos and \Posa
\cite{ErdosP65} proved that for every graph $G$, either $G$ contains
$k$ node-disjoint cycles, or there is a set $X$ of $O(k\log k)$ nodes,
whose removal from $G$ makes the graph acyclic. More generally, a
family $\mF$ of graphs is said to satisfy the \EP property, iff there
is an integer-valued function $f$, such that for every graph $G$,
either $G$ contains $k$ disjoint subgraphs isomorphic to members of
$\mF$, or there is a set $S$ of $f(k)$ nodes, such that $G-S$ contains
no subgraph isomorphic to a member of $\mF$. In other words, $S$ is a
cover, or a hitting set, for $\mF$ in $G$.  \EP-type results provide
relationships between integral covering and packing problems, and are
closely related to fractional covering problems and the integrality
gaps of the corresponding LP relaxations.

As an illustrative example for the \EP-type results, let $\fset_m$
denote the family of all cycles of length $0$ modulo
$m$. Thomassen~\cite{Thomassen88} has proved an \EP-type result for
$\fset_m$, by showing that for each graph $G$, either $G$ contains $k$
disjoint copies of cycles from $\fset_m$, or there is a subset $S$ of
$f(k)$ vertices, whose removal disconnects all such cycles in $G$
(here, $f(k)=2^{2^{O(k)}}$, and $m$ is considered to be a
constant). The proof consists of two steps. In the first step, a
simple inductive argument is used to show that for any graph $G$ of
treewidth at most $w$, either $G$ contains $k$ disjoint copies of
cycles from $\fset_m$, or there is a subset $S$ of $O(kw)$ vertices,
whose removal from $G$ disconnects all such cycles. The second step is
to show that if $G$ has treewidth at least some value $g(k)$, then it
must contain $k$ disjoint copies of cycles from $\fset_m$. Combining
these two steps together, we obtain that $f(k)=O(k\cdot g(k))$. The
second step uses Theorem~\ref{thm:RST-grid} to show that, if
$\tw(G)\geq g(k)=2^{m^{O(k)}}$, then $G$ contains a grid minor of size
$k(2m)^{2k-1}\times k(2m)^{2k-1}$. This grid minor is then in turn
used to find $k$ disjoint copies of cycles from $\fset_m$ in $G$,
giving $f(k)=2^{m^{O(k)}}$.

Using Theorem~\ref{thm:main1-intro}, we can significantly strengthen
this result, and obtain $f(k)=\tilde O(k)$, as
follows.\footnote{Throughout the paper we use $\tilde{O}$ notation to
  suppress polylogarithmic factors.} Assume first that we are given
any graph $G$, with $\tw(G)\geq f'(m)k\poly\log k$, where $f'(m)$ is
some function of $m$. Then, using Theorem~\ref{thm:main1-intro}, we
can partition $G$ into $k$ vertex-disjoint subgraphs of treewidth at
least $f'(m)$ each. Using known techniques (such as, e.g.,
Theorem~\ref{thm:RST-grid}), we can then show that each such subgraph
must contain a copy of a cycle from $\fset_m$. Therefore, if
$\tw(G)\geq f'(m)k\poly\log k$, then $G$ contains $k$ disjoint copies
of cycles from $\fset_m$. Combining this with Step 1 of the algorithm
of Thomassen, we conclude that every graph $G$ either contains $k$
copies of cycles from $\fset_m$, or there is a subset $S$ of
$f(k)=\tilde O(k^2)$ vertices, whose removal from $G$ disconnects all
such cycles; a stronger bound of $f(k) = \tilde O(k)$ can be obtained
by refining the Step 1 argument using a divide and conquer analysis \cite{FominST11} (see Lemma~\ref{lem:ep-improved} in Section~\ref{sec:apps}).

There is a large body of work in graph theory and combinatorics on
\EP-type results. Several of these rely on the \GM Theorem, and
consequently the function $f(k)$ is shown to be exponential (or even
worse) in $k$. Some fundamental results in this area can be improved
to obtain a bound polynomial in $k$, using
Theorem~\ref{thm:main1-intro} and the general framework outlined
above. For example, Robertson and Seymour \cite{RS-grid} derived the
following as an important consequence of the \GM Theorem. Given any
fixed graph $H$, let $\mF(H)$ be the family of all graphs that contain
$H$ as a minor. Then $\mF(H)$ has the \EP property iff $H$ is
planar. However, the bound they obtained for $f(k)$ is exponential in
$k$. Using the above general framework, we can show that $f(k)=O(k
\cdot \polylog(k))$ for any fixed $H$.

\smallskip {\em Improved running times for \FPT.} The theory of
bidimensionality \cite{DemaineH-survey} is a powerful methodology in
the design of fixed-parameter tractable (FPT) algorithms. It led to 
sub-exponential (in the parameter $k$) time FPT algorithms for
bidimensional parameters (formally defined in Section~\ref{sec:apps})
in planar graphs, and more generally graphs that exclude a fixed graph
$H$ as a minor.
The theory is based on the \GM
Theorem.  However, in general graphs, the weak bounds of the \GM
Theorem meant that one could only derive FPT algorithms with running
time of the form $2^{2^{O(k^{2.5})}}n^{O(1)}$, as shown by Demaine and
Hajiaghayi \cite{DemaineH07}. Our results lead to 
algorithms with running times of the form $2^{k~\polylog(k)}n^{O(1)}$ for the
same class of problems as in \cite{DemaineH07}. Thus, one can obtain
FPT algorithms for a large class of problems in general graphs via a
generic methodology, where the running time has a singly-exponential
dependence on the parameter $k$.

\smallskip
The thrust of this paper is to prove Theorems~\ref{thm:main1-intro} and
\ref{thm:main2-intro}, and to highlight their applicability as general
tools. The applications described in Section~\ref{sec:apps} are of
that flavor; we have not attempted to examine specific problems in
depth. We believe that the theorems, and the technical ideas
in their proofs, will have further applications.  

\smallskip
\noindent{\bf Overview of techniques and discussion.}
A significant contribution of this paper is the formulation of the
decomposition theorems for treewidth, and identifying their
applications. The main new and non-trivial technical contribution is
the proof of Theorem~\ref{thm:main1-intro}. The proof of
Theorem~\ref{thm:main2-intro} is similar in spirit to the recent work
of \cite{Chuzhoy11} and \cite{ChekuriE13}, who obtained a special case
of Theorem~\ref{thm:main2-intro} with $h=\poly\log k$, and used it to
design algorithms for low-congestion routing in undirected graphs. We
note that Theorem~\ref{thm:main1-intro} gives a substantially
different tradeoff between the parameters $h,r$ and $k$, when compared
to Theorem~\ref{thm:main2-intro}, and leads to the improved results
for the two applications we mentioned earlier. Its proof uses new
ingredients with a connection to decomposing expanders as explained
below.

\smallskip
\noindent
{\bf Contracted graph, well-linked decomposition, and expanders:}
The three key technical ingredients in the proof of
Theorem~\ref{thm:main1-intro} are in the title of the paragraph.  To
illustrate some key ideas we first consider how one may prove
Theorem~\ref{thm:main1-intro} if $G$ is an $n$-vertex constant-degree
expander, which has treewidth $\Omega(n)$. At a high level, one can
achieve this as follows. We can take $h$ disjoint copies of an
expander with $\Omega(r)$ nodes each (the expansion certifies that
treewidth of each copy is $r$), and ``embed'' them into $G$ in a
vertex-disjoint fashion. This is roughly possible, modulo various
non-trivial technical issues, using short-path vertex-disjoint routing
in expanders \cite{LR}. Now consider a general graph $G$. For instance
it can be a planar graph on $n$ nodes with treewidth $O(\sqrt{n})$;
note that the ratio of treewidth to the number of nodes is very
different than that in an expander. Here we employ a different
strategy, where we cut along a small separator and retain large
treewidth on both sides and apply this iteratively until we obtain the
desired number of subgraphs. The non-trivial part of the proof is to
be able to handle these different scenarios. Another technical
difficulty is the following. Treewidth of a graph is a global
parameter and there can be portions of the graph that can be removed
without changing the treewidth. It is not easy to cleanly characterize
the minimal subgraph of $G$ that has roughly the same treewidth as
that of $G$.  A key technical ingredient here is borrowed from
previous work on graph decompositions \cite{Raecke,Chuzhoy11}, namely,
the notion of a contracted graph. The contracted graph tries to
achieve this minimality, by contracting portions of the graph that
satisfy the following technical condition: they have a small boundary
and the boundary is well-linked with respect to the contracted
portion.  Finally, a recurring technical ingredient is the notion of a
well-linked decomposition. This allows us to remove a small number of
edges while ensuring that the remaining pieces have good conductance.
This high-level clustering idea has been crucial in many applications.


\smallskip
\noindent {\bf Related work on grid-like minors and (perfect) brambles:}
An important ingredient in the decomposition results is a need to
certify that the treewidth of a given graph is large, say at least
$r$. Interestingly, despite being NP-Hard to compute, the treewidth of
a graph $G$ has an exact min-max formula via the bramble number
\cite{ST-BN} (see Section \ref{sec:prelim}). However, Grohe and Marx
\cite{GroheM09} have shown that there are graphs $G$ (in fact
expanders) for which a polynomial-sized bramble can only certify that
treewidth of $G$ is $\Omega(\sqrt{k})$ where $k=\tw(G)$; certifying
that $G$ has larger treewidth would require super-polynomial sized
brambles.  Kreutzer and Tamari \cite{KreutzerT10}, building on
\cite{GroheM09}, gave efficient algorithms to construct brambles of
order $\tilde{\Omega}(\sqrt{k})$. They also gave efficient algorithms
to compute ``grid-like'' minors introduced by Reed and Wood
\cite{ReedW-grid} where it is shown that $G$ has a grid-like minor of
size $\ell$ as long as $\tw(G) = \Omega(\ell^4 \sqrt{\log \ell})$.  In
some applications it is feasible to use a grid-like minor in place of
a grid and obtain improved results.  Kreutzer and Tamari
\cite{KreutzerT10} used them to define perfect brambles and gave a
meta-theorem to obtain FPT algorithms, for a subclass of problems
considered in \cite{DemaineH07}, with a single-exponential dependence
on the parameter $k$. Our approach in this paper is different, and in
a sense orthogonal, as we explain below.

First, a grid-like minor is a single connected structure that does not
allow for a decomposition into disjoint grid-like minors. This
limitation means one needs a global argument to show that a grid-like
minor of a certain size implies a lower bound on some parameter of
interest. In contrast, our theorems are specifically tailored to
decompose the graph and then apply a local argument in each subgraph,
typically to prove that the parameter is at least one in each
subgraph. The advantage of our approach is that it is agnostic to how
one proves a lower bound in each subgraph; we could use the \GM
Theorem or the more efficient grid-like minor theorem in each
subgraph. Kreutzer and Tazari~\cite{KreutzerT10} derive efficient FPT
algorithms for a subclass of problems considered in \cite{DemaineH07}
where the class is essentially defined as those problems for which one
can use a grid-like minor in place of a grid in the global sense
described above. In contrast, we can generically handle all the
problems considered in \cite{DemaineH07} as explained in
Section~\ref{sec:apps}.

Second, we discuss the efficiency gains possible via our approach.  It
is well-known that an $\alpha$-approximation for sparse vertex
separators gives an $O(\alpha)$-approximation for treewidth. Feige
\etal \cite{FeigeHL05} obtain an $O(\sqrt{\log \tw(G)})$-approximation
for treewidth. Thus we can efficiently certify treewidth to within a
much better factor via separators than with brambles. More explicitly,
well-linked sets provide a compact certificate for treewidth;
informally, a set of vertices $X$ is well-linked in $G$ if there are
no small separators for $X$ --- see Section \ref{sec:prelim} for formal
definitions. The tradeoffs we obtain through well-linked sets are
stronger than via brambles.  In particular, the FPT algorithms that we
obtain have a running time $2^{k~\polylog(k)}n^{O(1)}$ where $k$ is
the parameter of interest. In contrast the algorithms obtained via
perfect brambles in \cite{KreutzerT10} have running times of the form
$2^{\poly(k)} n^{O(1)}$ where the polynomial is incurred due to the inefficiency
in the relationship between treewidth and the size of a grid-like
minor. Although the precise dependence on $k$ depends on the parameter
of interest, the current bounds require at least a quadratic dependence
on $k$.

\smallskip {\bf Organization:} \iffull
Sections~\ref{sec: proof of first main theorem} and \ref{sec: proof of
  second main theorem} contain the proofs of
Theorem~\ref{thm:main1-intro} and Theorem~\ref{thm:main2-intro}
respectively. Section~\ref{sec:apps} describes the applications; it
relies only on the statement of Theorem~\ref{thm:main1-intro} and can
be read independently.  \fi
\label{--------------------------------------------Prelims-------------------------------------}

\section{Preliminaries and Notation}
\label{sec:prelim}
Given a graph $G$ and a set of vertices $A$, we denote by
$\out_G(A)$ the set of edges with exactly one end point in $A$
and by $E_G(A)$ the set of edges with both end points in $A$.
For disjoint sets of vertices $A,B$ the set of edges with
one end point in $A$ and the other in $B$ is denoted by $E_G(A,B)$.
When clear from context, we omit the subscript $G$.  All
logarithms are to the base of $2$. 
We use the following
simple claim several times\iffull , and its proof appears in the Appendix\fi.

\begin{claim}\label{claim: simple partition} Let
  $\set{x_1,\ldots,x_n}$ be a set of non-negative integers, with
  $\sum_{i}x_i=N$, and $x_i\leq 2N/3$ for all $i$. Then we can
  efficiently compute a partition $(A,B)$ of $\set{1,\ldots,n}$, such
  that $\sum_{i\in A}x_i,\sum_{i\in B}x_i\geq N/3$.
\end{claim}

\noindent {\bf Graph partitioning.}
Suppose we are given any graph $G=(V,E)$ with a set $T$ of vertices
called terminals.  Given any partition $(S,\nots)$ of $V(G)$, the
\emph{sparsity} of the cut $(S,\nots)$ with respect to $T$ is
$\Phi(S,\nots)=\frac{|E(S,\nots)|}{\min\set{|T\cap S|,|T\cap
    \nots|}}$. The \emph{conductance} of the cut $(S,\nots)$ is
$\Psi(S,\nots)=\frac{|E(S,\nots)|}{\min\set{|E(S)|,|E(\nots)|}}$.  We
then denote: $\Phi(G)=\min_{S\subset V}\set{\Phi(S,\nots)}$, and
$\Psi(G)=\min_{S\subset V}\set{\Psi(S,\nots)}$.  Arora, Rao and
Vazirani~\cite{ARV} showed an efficient algorithm that, given a graph $G$ with a
set $T$ of $k$ terminals, produces a cut $(S,\nots)$ with
$\Phi(S,\nots)\leq \alphasc(k)\cdot \Phi(G)$, where
$\alphasc(k)=O(\sqrt{\log k})$. Their algorithm can also be used to
find a cut $(S,\nots)$ with $\Psi(S,\nots)\leq \alphasc(m)\cdot
\Psi(G)$, where $m=|E(G)|$. We denote this algorithm by $\algsc$, and
its approximation factor by $\alphasc$ from now on.

\smallskip
\iffull
\noindent {\bf Well-linkedness, bramble number and treewidth.} \fi
The treewidth of a graph $G=(V,E)$ is typically defined via tree
decompositions.  A tree-decomposition for $G$ consists of a tree
$T=(V(T),E(T))$ and a collection of sets $\{X_v \subseteq V\}_{v \in
  V(T)}$ called bags, such that the following two properties are
satisfied: (i) for each edge $ab \in E$, there is some node $v \in
V(T)$ with both $a,b \in X_v$ and (ii) for each vertex $a \in V$, the
set of all nodes of $T$ whose bags contain $a$ form a non-empty 
(connected) subtree of $T$. The {\em width} of a given tree decomposition is
$\max_{v \in V(T)} |X_v| - 1$, and the treewidth of a graph $G$,
denoted by $\tw(G)$, is the width of a minimum-width tree
decomposition for $G$.

It is convenient to work with well-linked sets instead of treewidth.
We describe the relationship between them after formally defining the
notion of well-linkedness that we require.

\begin{definition}
  We say that a set $T$ of vertices is
  $\alpha$-well-linked\footnote{This notion of well-linkedness is
    based on edge-cuts and we distinguish it from node-well-linkedness
    that is directly related to treewidth. For technical reasons it is
    easier to work with edge-cuts and hence we use the term well-linked to mean
    edge-well-linkedness, and explicitly use the term
    node-well-linkedness when necessary.} in $G$, iff for any
  partition $(A,B)$ of the vertices of $G$ into two subsets,
  $|E(A,B)|\geq \alpha\cdot \min\set{|A\cap T|,|B\cap T|}$.
\end{definition}

\begin{definition}
  We say that a set $T$ of vertices is \emph{node-well-linked} in $G$,
  iff for any pair $(T_1,T_2)$ of equal-sized subsets of $T$, there is
  a collection $\pset$ of $|T_1|$ {\bf node-disjoint} paths,
  connecting the vertices of $T_1$ to the vertices of $T_2$. (Note
  that $T_1$, $T_2$ are not necessarily disjoint, and we allow empty
  paths).
\end{definition}

\begin{lemma}[Reed \cite{Reed-chapter}]
  \label{lem:tw-wl}
  Let $k$ be the size of the largest node-well-linked set in $G$. Then
  $k \le \tw(G) \le 4k$.
\end{lemma}

\iffull

We also use the notion of brambles, defined below. Brambles will help
us relate the different notions of well-linkedness, and treewidth to
each other.

\begin{definition}
  A bramble in a graph $G$ is a collection
  $\bset=\set{G_1,\ldots,G_r}$ of connected sub-graphs of $G$, where
  for every pair $G_i,G_j$ of the subgraphs, either $G_i$ and $G_j$
  share at least one vertex, or there is an edge $e=(u,v)$ with $u\in
  G_i$, $v\in G_j$. We say that a set $S$ of vertices is a hitting set
  for the bramble $\bset$, iff for each $G_i\in \bset$, $S\cap
  V(G_i)\neq \emptyset$. The order of the bramble $\bset$ is the
  minimum size of any hitting set $S$ for $\bset$. The bramble number
  of $G$, $\BN(G)$, is the maximum order of any bramble in $G$.
\end{definition}

\begin{theorem}\cite{ST-BN}
For every graph $G$, $\tw(G)=\BN(G)-1$.
\end{theorem}

\fi

We then obtain the following simple corollary, \iffull whose proof
appears in the Appendix.\fi

\begin{corollary}\label{cor: from well-linkedness to treewidth}
  Let $G$ be any graph with maximum vertex degree at most $\Delta$,
  and let $T$ be any subset of vertices, such that $T$ is
  $\alpha$-well-linked in $G$, for some $0<\alpha<1$. Then $\tw(G)\geq
  \frac{\alpha \cdot |T|}{3\Delta}-1$.
\end{corollary}

Lemma~\ref{lem:tw-wl} guarantees that any graph $G$ of treewidth $k$
contains a set $X$ of $\Omega(k)$ vertices, that is node-well-linked
in $G$. Kreutzer and Tazari~\cite{KreutzerT10} give a constructive
version of this lemma, obtaining a set $X$ with slightly weaker
properties. Lemma~\ref{lem: find well-linked set} below rephrases, in
terms convenient to us, Lemma 3.7 in~\cite{KreutzerT10} .
 
 \begin{lemma}\label{lem: find well-linked set} There is an
 efficient algorithm, that, given a graph $G$ of treewidth
   $k$, finds a set $X$ of $\Omega(k)$ vertices, such that $X$ is
   $\alpha^*=\Omega(1/\log k)$-well-linked in $G$. Moreover, for any
   partition $(X_1,X_2)$ of $X$ into two equal-sized subsets, there is
   a collection $\pset$ of paths connecting every vertex of $X_1$ to a
   distinct vertex of $X_2$, such that every vertex of $G$
   participates in at most $1/\alpha^*$ paths in $\pset$.
 \end{lemma}

\noindent{\bf Well-linked decompositions.}
Let $S$ be any subset of vertices in $G$. We say that $S$ is
$\alpha$-good\footnote{The same property was called "bandwidth
  property" in~\cite{Raecke}, and in~\cite{Chuzhoy11,ChuzhoyL12}, set
  $S$ with this property was called $\alpha$-well-linked. We choose
  this notation to avoid confusion with other notions of
  well-linkedness used in this paper.}, iff for any partition $(A,B)$
of $S$, $|E(A,B)|\geq \alpha\cdot \min\set{|\out(A)\cap
  \out(S)|,|\out(B)\cap \out(S)|}$. An equivalent definition is
as follows. Start with graph $G$ and subdivide each edge $e \in \out_G(S)$ by a vertex $t_e$.
Let $\tset_S=\set{t_e\mid e\in \out_G(S)}$ be the set of these new vertices, and let $H$ be the sub-graph of the resulting graph, induced by $S\cup \tset_S$.
 Then $S$ is $\alpha$-good in $G$
iff  $\tset_S$ is $\alpha$-well-linked in $H$.

A set $D:\out(S)\times \out(S)\rightarrow\reals^+$ of demands defines,
for every pair $e,e'\in \out(S)$, a demand $D(e,e')$. We say that $D$
is a $c$-restricted set of demands, iff for every $e\in \out(S)$,
$\sum_{e'\in \out(S)}D(e,e')\leq c$. Assume that $S$ is an
$\alpha$-good subset of vertices in $G$. From the duality of cuts and
flows, and from the known bounds on the flow-cut gap in undirected
graphs \cite{LLR}, if $D$ is any set of $c$-restricted demands over
$\out(S)$, then it can be fractionally routed inside $G[S]$ with
edge-congestion at most $O(c\log k'/\alpha)$, where $k'=|\out(S)|$.

The following theorem, in its many variations, (sometimes under the
name of "well-linked decomposition") has been used extensively in
routing and graph decomposition (see
e.g. \cite{Raecke,ANF,CKS,RaoZhou,Andrews,Chuzhoy11,ChuzhoyL12,ChekuriE13}). \iffull For completeness,
the proof appears in Appendix.\fi

\begin{theorem}\label{thm: well-linked-decomposition}
  Let $S$ be any subset of vertices of $G$, with $|\out(S)|=k'$, and
  let $0<\alpha<\frac{1}{8\alphasc(k')\cdot \log k'}$ be a
  parameter. Then there is an efficient algorithm to compute a
  partition $\wset$ of $S$, such that for each $W\in \wset$,
  $|\out(W)|\leq k'$ and $W$ is $\alpha$-good. Moreover, $\sum_{W\in
    \wset}|\out(W)|\leq k'(1+16\alpha\cdot \alphasc(k')\cdot \log
  k')=k'(1+O(\alpha\log^{3/2}k'))$. The parameter $\alphasc(k')$ can be
  set to $1$ if  the efficiency of the algorithm is not relevant.
\end{theorem}

\noindent {\bf Pre-processing to reduce maximum degree.}
Let $G$ be any graph with $\tw(G) = k$. The proofs of
Theorems~\ref{thm:main1-intro} and \ref{thm:main2-intro} work with
edge-well-linked sets instead of the node-well-linked ones. In order
to be able to translate between both types of well-linkedness and the
treewidth, we need to reduce the maximum vertex degree of the input
graph $G$.  Using the cut-matching game of Khandekar, Rao and
Vazirani~\cite{KRV}, one can reduce the maximum vertex degree to
$O(\log^3 k)$, while only losing a $\poly\log k$ factor in the
treewidth, as was noted in \cite{ChekuriE13} (see Remark 2.2).  We
state the theorem formally below. A brief proof sketch appears in
\iffull the Appendix for
completeness.\fi

\begin{theorem}\label{thm:log-degree}
  Let $G$ be any graph with treewidth $k$. Then there is an efficient randomized algorithm to compute a
  subgraph $G'$ of $G$, with maximum vertex degree at most $O(\log^3 k)$ 
  such that $\tw(G')=\Omega(k/\log^6 k)$. 
\end{theorem}

\begin{remark}
  In fact a stronger result, giving a constant bound on the maximum
  degree follows from the expander embedding result in
  \cite{ChekuriE13}. However, the bound on the treewidth guaranteed is
  worse than in the preceding theorem by a (large) polylogarithmic
  factor. For our algorithms, the polylogarithmic bound on the degree
  guaranteed by Theorem~\ref{thm:log-degree} is sufficient.
\end{remark}

\label{--------------------------------------------sec Proof of main thm 1------------------------------------------}
\iffull
\section{Proof of Theorem~\ref{thm:main1-intro}}
\fi
\label{sec: proof of first main theorem}

We start with a graph $G$ whose treewidth is at least $k$. For our
algorithm, we need to know the value of the treewidth of $G$, instead
of the lower bound on it. We can compute the treewidth of $G$
approximately, to within an $O(\log(\tw(G)))$-factor, using the
algorithm of Amir~\cite{Amir10}.  Therefore, we assume that we are
given a value $k'\geq k$, such that $\Omega(k'/\log k')\leq \tw(G)\leq
k'$. We then apply Theorem~\ref{thm:log-degree}, to obtain a subgraph
$G'$ of $G$ of maximum vertex degree $\Delta=O(\log^3k')$ and
treewidth $\Omega(k'/\log^7 k')$.  Using Lemma~\ref{lem: find
  well-linked set}, we compute a subset $T$ of $\Omega(k'/\log^7k')$
vertices, such that $T$ is $\Omega(1/\log k')$-well-linked in
$G'$.\footnote{The bounds we claim here are somewhat loose although they
  do not qualitatively affect our theorems. For instance
  \cite{FeigeHL05} gives an $O(\sqrt{\log(\tw(G))})$-approximation for
  treewidth which improves the bound in \cite{Amir10}. 
  }

In order to simplify the notation, we denote $G'$ by $G$ and $|T|$ by
$k$ from now on.  From the above discussion, $\tw(G)\leq ck\log^7k$
for some constant $c$, $T$ is $\Omega(1/\log k)$-well-linked in $G$,
and the maximum vertex degree in $G$ is $\Delta=O(\log^3k)$; we define
the parameter $\alpha^*$ to be $\Omega(1/\log k)$ which is the
well-linkedness guarantee given by Lemma~\ref{lem: find well-linked
  set}.  It is now enough to find a collection $G_1,\ldots,G_h$ of
vertex-disjoint subgraphs of $G$, such that $\tw(G_i)\geq r$ for each
$i$.  We use the parameter $r'=c'\Delta^2 r\log^{11}k$, where $c'$ is
a sufficiently large constant. We assume without loss of generality
that $k$ is large enough, so, for example, $k\geq c''r\log^{30}k$,
where $c''$ is a large enough constant.  We also assume without loss
of generality that $G$ is connected.

\begin{definition}
  We say that a subset $S$ of vertices in $G$ is an acceptable
  cluster, iff $|\out(S)|\leq r'$, $|S\cap T|\leq |T|/2$, and $S$ is
  $\alphaWL$-good, for $\alphaWL=\frac{1}{256\alphasc(k)\log
    k}=\Theta\left (\frac{1}{\log^{1.5}k}\right)$.
\end{definition}

Notice that since the maximum vertex degree in $G$ is bounded by
$\Delta<r'$, if $S$ consists of a single vertex, then it is an
acceptable cluster.  Given any partition $\cset$ of the vertices of
$G$ into acceptable clusters, we let $H_{\cset}$ be the
\emph{contracted graph} associated with $\cset$. Graph $H_{\cset}$ is
obtained from $G$ by contracting every cluster $C\in\cset$ into a
single vertex $v_C$, that we refer to as a super-node. We delete
self-loops, but leave parallel edges. Notice that the maximum vertex
degree in $H_{\cset}$ is bounded by $r'$. We denote by $\phi(\cset)$
the total number of edges in $H_{\cset}$.  Below is a simple
observation that follows from the $\alpha^*$-well-linkedness of $T$ in
$G$. 

\begin{observation}\label{obs: many edges}
  Let $\cset$ be any partition of the vertices of $G$ into acceptable
  clusters. Then $\phi(\cset)\geq \alpha^* k/3$.
\end{observation}

\iffull

\begin{proof}
  We use Claim~\ref{claim: simple partition} to find a partition
  $(\aset,\bset)$ of $\cset$, such that $\sum_{C\in \aset}|T\cap
  C|,\sum_{C\in \bset}|T\cap C|\geq |T|/3$. We then set
  $A=\bigcup_{C\in \aset}C$, $B=\bigcup_{C\in \bset}C$. Since $T$ is
  $\alpha^*$-well-linked in $G$,  $|E(A,B)|\geq \alpha^* |T|/3=\alpha^* k/3$.
  Since the edges of $E(A,B)$ also belong to $H_{\cset}$, the claim follows.
\end{proof}

\fi

Throughout the algorithm, we maintain a partition $\cset$ of $V(G)$
into acceptable clusters. At the beginning, $\cset=\set{\set{v}\mid
  v\in V(G)}$.  We then perform a number of iterations. In each
iteration, we either compute a partition of $G$ into $h$ disjoint
sub-graphs, of treewidth at least $r$ each, or find a new partition
$\cset'$ of $V(G)$ into acceptable clusters, such that
$\phi(\cset')\leq \phi(\cset)-1$. The execution of each iteration is
summarized in the following theorem.

\begin{theorem}\label{thm: iteration}
  There is an efficient algorithm, that, given a partition $\cset$ of
  $V(G)$ into acceptable clusters, either computes a partition of $G$
  into $h$ disjoint subgraphs of treewidth at least $r$ each, or
  returns a new partition $\cset'$ of $V(G)$ into acceptable clusters,
  such that $\phi(\cset')\leq \phi(\cset)-1$.
\end{theorem}

Clearly, after applying Theorem~\ref{thm: iteration} at most $|E(G)|$
times, we obtain a partition of $G$ into $h$ disjoint subgraphs of
treewidth at least $r$ each. From now on we focus on proving
Theorem~\ref{thm: iteration}. Given a current partition $\cset$ of
$V(G)$ into acceptable clusters, let $H$ denote the corresponding
contracted graph.  We denote $n=|V(H)|$, $m=|E(H)|$. Notice that from
Observation~\ref{obs: many edges}, $m\geq \alpha^*k/3$.  We now
consider two cases, and prove Theorem~\ref{thm:main1-intro} separately
for each of them. The first case is when $n\geq k^5$.

\subsection{Case 1: $n\geq k^5$}




We note that $n$ is large when compared to the treewidth and hence we
expect the graph $H$ should have low expansion. Otherwise, we get a
contradiction by showing that $\tw(G) > ck\log^7k$. The proof strategy
in the low-expansion regime is to repeatedly decompose along balanced
partitions to obtain $h$ subgraphs with treewidth at least $r$ each.

Let $z=k^5$. The algorithm first chooses an arbitrary subset $Z$ of
$z$ vertices from $H$ that remains fixed throughout the
algorithm. Suppose we are given any subset $S$ of vertices of $H$. We
say that a partition $(A,B)$ of $S$ is \emph{$\gamma$-balanced} (with
respect to $Z$), iff $\min\{|A\cap Z|,|B\cap Z|\} \geq \gamma |S\cap
Z|$.  We say that it is \emph{balanced} iff it is $\gamma$-balanced
for $\gamma=\frac 1 4$. The following claim is central to the proof of
the theorem in Case 1.

\begin{claim}\label{claim: small balanced cut}
  Let $S$ be any subset of vertices in $H$ with $|S\cap Z|>100$, and
  let $(A,B)$ a balanced partition of $S$ (with respect to $Z$), minimizing
  $|E_H(A,B)|$. Then $|E_H(A,B)|\leq k^2$.
\end{claim}

\iffull
\begin{proof}
  To simplify notation, we denote $H[S]$ by $H'$.  Assume that the
  claim is not true, and assume without loss of generality that
  $|A\cap Z|\geq |B\cap Z|$. We claim that the set $A$ of vertices is
  $1$-good in $H'$. Indeed, assume otherwise. Then there is a
  partition $(X,Y)$ of $A$, with

\[
|E_{H'}(X,Y)|< \min\set{|\out_{H'}(X)\cap \out_{H'}(A)|,|\out_{H'}(Y)\cap \out_{H'}(A)|}.
\]
 
Assume without loss of generality that $|X\cap Z|\geq |Y\cap Z|$. We claim that
$(X,B\cup Y)$ is a balanced partition for $S$, with $|E_H(X,B\cup
Y)|<|E_H(A,B)|$, contradicting the minimality of $|E(A,B)|$. To see
that $(X,B\cup Y)$ is a balanced partition, observe that $|(B\cup Y)\cap
Z|\geq \frac{|S\cap Z|}{4}$ since $(A,B)$ is a balanced partition, and
$|X\cap Z|\geq \half |A\cap Z|\geq \frac 1 4 |A\cap Z|$ from our
assumptions about $X$ and $A$. Finally, observe that

\[\begin{split}
|E_H(X,B\cup Y)|&=
|E_{H'}(X,B\cup Y)|
\\&=|E_{H'}(X,Y)|+|E_{H'}(X,B)|
\\&<|E_{H'}(Y,B)|+|E_{H'}(X,B)|
\\&=|E_{H'}(A,B)|=|E_{H}(A,B)|,\end{split}\]

contradicting the minimality of $|E_{H}(A,B)|$. We conclude that the
set $A$ of vertices is $1$-good in $H'$. Let $\Gamma$ be the subset of
vertices of $A$ that serve as endpoints to edges in $\out_{H'}(A)$,
that is, $\Gamma=\set{v\in A\mid \exists e=(u,v)\in
  \out_{H'}(A)}$. Then $|\Gamma|\geq
\frac{|E_{H'}(A,B)|}{r'}\geq\frac{k^2}{r'}$, since the degrees of
vertices in $H$ are bounded by $r'$. It is also easy to see that
$\Gamma$ is $1$-well-linked in the graph $H[A]$, since $A$ is a
$1$-good set in $H'$.

Finally, let $\Gamma'$ be a subset of $|\Gamma|$ vertices in the
original graph $G$, obtained as follows. For each super-node $v_C\in
\Gamma$, we select an arbitrary vertex $u$ on the boundary of $C$
(that is, $u\in C$, and it is an endpoint of some edge in $\out(C)$),
and add it to $\Gamma'$. We claim that the set $\Gamma'$ of vertices
is $\Omega\left (\frac{\alphaWL}{\log k}\right )$-well-linked in
$G$. Indeed, let $(X,Y)$ be any partition of the vertices of $G$, with
$\Gamma'_X=\Gamma'\cap X, \Gamma'_Y=\Gamma'\cap Y$, and assume without
loss of generality that $|\Gamma'_X|\leq |\Gamma'_Y|$. We need to show
that $|E(X,Y)|\geq \Omega\left(\frac{\alphaWL|\Gamma'_X|}{\log
    k}\right )$. In order to show this, it is enough to show that
there is a flow $F$, where the vertices in $\Gamma'_X$ send one flow
unit each to the vertices of $\Gamma'_Y$, and the total
edge-congestion caused by this flow is at most $O(\log
k/\alphaWL)$. Let $(\Gamma_X,\Gamma_Y)$ be the partition of $\Gamma$
induced by the partition $(\Gamma_X',\Gamma_Y')$ of $\Gamma'$. Since
the set $\Gamma$ of vertices is $1$-well-linked in $H$, there is a
flow $F'$ in $H$, where every vertex in $\Gamma_X$ sends $1$ flow unit
towards the vertices in $\Gamma_Y$, every vertex in $\Gamma_Y$
receives at most one flow unit, and the edge-congestion is at most
$1$. We now extend this flow $F'$ to obtain the desired flow $F$ in
the graph $G$. In order to do so, we need to specify how the flow is
routed across each cluster $C\in \cset$. For each such cluster $C$,
flow $F$ defines a set $D_C$ of $2$-restricted demands over $\out(C)$
(the factor $2$ comes from both the flow routed across the cluster,
and the flow that originates or terminates in it). Since $C$ is
$\alphaWL$-well-linked, this set $D_C$ of demands can be routed inside
$C$ with congestion at most $O(\log r'/\alphaWL)\leq O(\log
k/\alphaWL)$. Concatenating the flow $F'$ with the resulting flows
inside each cluster $C\in \cset$ gives the desired flow $F$. We
conclude that we have obtained a set $\Gamma$ of at least
$\frac{k^2}{r'}$ vertices in $G$, such that $\Gamma$ is
$\Omega(\alphaWL/\log k)$-well-linked. From Corollary~\ref{cor: from
  well-linkedness to treewidth}, it follows that $\tw(G)\geq
\Omega\left (\frac{\alphaWL(k)\cdot k^2}{r'\cdot \Delta\log k}\right
)=\Omega\left(\frac{k^2}{r\cdot \log^{22.5} k}\right )>ck\log^7k$,
since we have assumed that $k\geq c''r\log^{30}k$ for a large enough
constant $c''$. This contradicts the fact that $\tw(G)\leq ck\log^7k$.
\end{proof}
\fi

We now show an algorithm to find the desired collection
$G_1,\ldots,G_h$ of subgraphs of $G$.  We use the algorithm $\algSC$
of Arora, Rao and Vazirani~\cite{ARV} to find a balanced partition of
a given set $S$ of vertices of $H$; the algorithm is applied to $H$
with $S \cap Z$ as the terminals. Given any such set $S$ of vertices,
the algorithm $\algSC$ returns a $\gammaARV$-balanced partition
$(A,B)$ of $S$, with $|E_H(A,B)|\leq \alphaSC(z)\cdot \opt$, where
$\opt$ is the smallest number of edges in any balanced partition, and
$\gammaARV$ is some constant. In particular, from Claim~\ref{claim:
  small balanced cut}, $|E(A,B)|\leq \alphaARV(z)\cdot k^2$, if
$|S\cap Z|\geq 100$.

 We
start with $\sset=\set{V(H)}$, and 
perform $h$
iterations. At the beginning of iteration $i$, set $\sset$ will
contain $i$ disjoint subsets of vertices of $H$. An iteration is
executed as follows. We select a set $S\in \sset$, maximizing
$|Z\cap S|$, and compute a $\gammaARV$-balanced partition $(A,B)$ of $S$,  using the algorithm \algSC. We then remove $S$ from $\sset$, and add $A$
and $B$ to it instead. Let $\sset=\set{X_1,\ldots,X_{h+1}}$ be the final
collection of sets after $h$ iterations.  From
Claim~\ref{claim: small balanced cut}, the increase in $\sum_{X\in
  \sset}|\out_H(X)|$ is bounded by $k^2\alphasc(z)$ in each iteration. Therefore,
throughout the algorithm, $\sum_{X\in \sset}|\out_H(X)|\leq k^2\alphasc(z)h$
holds. In the following observation,  we show that for each $X_i\in
\sset$, $|X_i\cap Z|\geq \frac{\gammaARV \cdot z}{2h}$.

\begin{observation}\label{obs: many edges in each cut} Consider some iteration $i$ of the algorithm. Let
  $\sset_i$ be the collection of vertex subsets at the
  beginning of iteration $i$, let $S\in \sset_i$ be the set that was
  selected in this iteration, and let $\sset_{i+1}$ be the set
  obtained after replacing $S$ with $A$ and $B$. Then
  $|A\cap Z|,|B\cap Z|\geq \gammaARV\cdot |S\cap Z|$, and for each $S'\in
  \sset_{i+1}$, $|S'\cap Z|\geq \frac{\gammaARV \cdot z}{2h}$.
\end{observation}

\iffull

\begin{proof}
  The proof is by induction on the number of iterations.  At the
  beginning of iteration $1$, $\sset_1=\set{V(H)}$, so the claim
  clearly holds for $\sset_1$. Assume now that the claim holds for
  iterations $1,\ldots,i-1$, and consider iteration $i$, where some
  set $S\in \sset_{i}$ is replaced by $A$ and $B$.  
  Since $(A,B)$ is a $\gammaARV$-balanced cut of $S$, and $|S|\geq \frac{\gammaARV \cdot z}{2h}>100$, it follows that $|A\cap Z|,|B\cap Z|\geq \gammaARV\cdot |S\cap Z|$.
  
  Since we have assumed that the claim holds for iterations
  $1,\ldots,i-1$, and we have shown that $|A\cap Z|,|B\cap Z|\geq \gammaARV\cdot |S\cap Z|$, it follows that the ratio $\max_{S'\in
    \sset_{i+1}}\set{|S'\cap Z|}/\min_{S'\in \sset_{i+1}}\set{|S'\cap Z|}\leq
  1/\gammaARV$. Therefore,  for each $S'\in
  \sset_{i+1}$, $|S'\cap Z|\geq \frac{\gammaARV \cdot z}{2h}$
\end{proof}
\fi

Among the sets $X_1,\ldots,X_{h+1}$, there can be at most one set
$X_i$, with $|T\cap \left (\bigcup_{v_C\in X_i}C\right )|>|T|/2$. We
assume without loss of generality that this set is $X_{h+1}$, and we will ignore it from
now on.  Consider now some set $X_i$, for $1\leq i\leq
h$. 
Since graph $H$ is connected, and $X_i$ contains at least $\frac{\gammasc\cdot z}{2h}$ vertices (the vertices of $X_i\cap Z$), while $|\out_H(X_i)|\leq k^2h\alphasc(z)$, it follows that $|E_H(X_i)|\geq \half\left (\frac{\gammasc\cdot z}{2h}-k^2h\alphasc(z)\right )\geq \frac{\gammasc\cdot z}{8h}>64|\out_H(X_i)|$, as $z=k^5$, and $k$ is large enough.

Let $X'_i$ be the subset of vertices obtained from $X_i$, by
un-contracting all super-nodes of $X_i$. Then $|E_G(X'_i)|\geq
|E_H(X_i)| \geq 64|\out_H(X_i)|=
64|\out_G(X'_i)|$.

Our next step is to compute a decomposition $\wset_i$ of $X'_i$ into $\alphaWL$-good
clusters, using Theorem~\ref{thm: well-linked-decomposition}. Notice that $k'=|\out_G(X'_i)|\leq k^2h\alphasc(z)\leq 5k^2h\alphasc(k)<k^4$ since $z=k^5$; therefore the choice of $\alpha_G=\frac{1}{256\alphasc (k)\log k}<\frac{1}{8\alphasc(k')\log k'}$ satisfies the conditions of the theorem.

Fix some $1\leq i\leq h$. Assume first that for every cluster $C_i\in \wset_i$, 
$|\out_G(C_i)|\leq r'$. Then we can obtain a new
partition $\cset'$ of $V(G)$ into acceptable clusters with
$\phi(\cset')\leq\phi(\cset)-1$, as follows. We add to $\cset'$ all
clusters $C\in \cset$ that are disjoint from $X'_i$, and we add all
clusters in $\wset_i$ to it as well. Clearly, the resulting partition
$\cset'$ consists of acceptable clusters only. We now show that
$\phi(\cset')\leq\phi(\cset)-1$.  Indeed,

\[\phi(\cset')\leq \phi(\cset)-|\out_H(X_i)|-|E_H(X_i)|+\sum_{R\in \wset_i}|\out_G(R)|\]

From the choice of $\alphaWL$, $\sum_{R\in \wset_i}|\out_G(R)|< 3
|\out_G(X'_i)| = 3|\out_H(X_i)|$ holds, while $|E_H(X_i)|\geq
64|\out_H(X_i)|$. Therefore, $\phi(\cset')\leq\phi(\cset)-1$.

Assume now that for each $1\leq i\leq h$, there is at least one
cluster $C_i\in \wset_i$ with $|\out_G(C_i)|\geq r'$. Let
$\set{C_1,\ldots,C_h}$ be the resulting collection of clusters, where
for each $i$, $C_i\in \wset_i$. For $1\leq i\leq h$, we now let
$G_i=G[C_i]$. It is easy to see that the graphs $G_1,\ldots,G_h$ are
vertex-disjoint. It now only remains to show that each graph $G_i$ has
treewidth at least $r$. Fix some $1\leq i\leq h$, and let
$\Gamma_i\subseteq C_i$ contain the endpoints of edges in
$\out_G(C_i)$, that is, $\Gamma_i=\set{v\in C_i\mid \exists e=(u,v)\in
  \out_G(C_i)}$. Then, since $C_i$ is an $\alphaWL$-good set of
vertices, $\Gamma_i$ is $\alphaWL$-well-linked in the graph
$G_i$. Moreover, $|\Gamma_i|\geq |\out(C_i)|/\Delta\geq
r'/\Delta$. From Corollary~\ref{cor: from well-linkedness to
  treewidth}, $\tw(G_i)\geq \frac{\alphaWL r'}{3\Delta^2}-1\geq r$.

\subsection{Case 2: $n<k^5$}
Since vertex degrees in $H$ are bounded by $r'$, $m=O(k^5r')=O(k^6)$.  The
algorithm for Case 2 consists of two phases. In the first phase, we
partition $V(H)$ into a number of disjoint subsets
$X_1,\ldots,X_{\ell}$, where, on the one hand, for each $X_i$, the
conductance of $H[X_i]$ is large, while, on the other hand,
$\sum_{i=1}^{\ell}|\out(X_i)|\leq |E(H)|/10$. We discard all clusters
$X_i$ with $|\out(X_i)|\geq |E(X_i)|/2$, denoting by $\xset$ the
collection of the remaining clusters, and show that $\sum_{X_i\in
  \xset}|E(X_i)|=\Omega(\alpha^*k)$. If any cluster $X\in \xset$ has
$|E(X)|\leq 2r'$, then we find a new partition $\cset'$ of the
vertices of $G$ into acceptable clusters, with $\phi(\cset')
\leq\phi(\cset)-1$.
Therefore, we can assume that for every cluster $X\in \xset$,
$|E(X)|>2r'$. We then proceed to the second phase.  Here, we take
advantage of the high conductance of each $X_i \in \xset$ to show that
$X_i$ can be partitioned into $h_i$ vertex-disjoint sub-graphs, such
that we can embed a large enough expander into each such subgraph. The
value $h_i$ is proportional to $|E(X_i)|$, and we ensure that
$\sum_{X_i\in \xset}h_i\geq h$ to get the desired number of
subgraphs. The embedding of the expander into each sub-graph is then
used as a certificate that this sub-graph (or more precisely, a
sub-graph of $G$ obtained after un-contracting the super-nodes) has 
large treewidth.

\paragraph{Phase 1} We use the following theorem, that allows us to
decompose any graph into a collection of high-conductance connected
components, by only removing a small fraction of the edges. A similar
procedure has been used in previous work, and can be proved using
standard graph decomposition techniques. \iffull The proof is deferred to the Appendix.\fi

\begin{theorem}\label{thm: cut into expanders}
Let $H$ be any connected $n$-vertex graph containing $m$ edges. Then there is an efficient algorithm to compute a partition $X_1,\ldots,X_{\ell}$ of the vertices of $H$, such that:
(i) for each $1\leq i\leq \ell$, the conductance of graph $H[X_i]$, $\Psi(H[X_i])\geq \frac{1}{160\alphaSC(m)\log m}$; and (ii) $ \sum_{i=1}^{\ell}|\out(X_i) |\leq m/10$.
\end{theorem}

The algorithm in phase 1 uses Theorem~\ref{thm: cut into expanders} to
partition the contracted graph $H$ into a collection
$\set{X_1,\ldots,X_{\ell}}$ of clusters. Recall that $m=|E(H)|$, and
$n=|V(H)|$. We are guaranteed that $\sum_{i=1}^{\ell}|E(X_i)|\geq
0.9m$ and $\sum_{i=1}^{\ell}|\out(X_i)| \leq 0.1m$ from Theorem~\ref{thm: cut   into expanders}.
 
Let $\xset'$ contain all clusters $X_i$ with $|\out(X_i)|\geq \half
|E(X_i)|$, and let $\xset$ contain all remaining clusters. Notice that 
$\sum_{X_i\in \xset'}|E(X_i)|\leq 2\sum_{X_i\in \xset'}|\out(X_i)|\leq
2\sum_{i=1}^{\ell}|\out(X_i)|\leq 0.2m$. Therefore, $\sum_{X_i\in
  \xset}|E(X_i)|\geq \half m\geq \frac{\alpha^* k} 6$ from Observation~\ref{obs:
  many edges}. From now on we only focus on clusters in $\xset$.

Assume first that there is some cluster $X_i\in \xset$, with
$|E(X_i)|\leq 2r'$. We claim that in this case, we can find a new
partition $\cset'$ of the vertices of $G$ into acceptable clusters,
with $\phi(\cset') \leq \phi(\cset)-1$. We
first need the following simple observation.

\begin{observation}\label{obs: small cut few terminals}
Assume that for some $X_i\in\xset$, $|E(X_i)|\leq 2r'$ holds. Let $X'=\bigcup_{v_C\in X_i}C$. Then $|X'\cap T|<|T|/2$.
\end{observation}

\iffull
\begin{proof}
Assume otherwise. Observe that since $X_i\in \xset$, $|\out(X_i)|\leq |E(X_i)|/2\leq r'$ must hold. Let $\cset'\subseteq \cset$ be the collection of all acceptable clusters, whose corresponding super-nodes belong to $X_i$. 
Let $C^*=\bigcup_{v_C\in V(H)\setminus X_i}C$, and let $\cset''=\cset'\cup\set{C^*}$. From our assumption, $|C^*\cap T|\leq |T|/2$, so we can use Claim~\ref{claim: simple partition} to partition the clusters in $\cset''$ into two subsets, $\aset$ and $\bset$, such that $\sum_{C\in \aset}|T\cap C|,\sum_{C\in \bset}|T\cap C|\geq |T|/3$. Let $A=\bigcup_{C\in \aset}C$ and let $B=\bigcup_{C\in \bset}C$. Then  $|E_G(A,B)|\geq \alpha^*|T|/3$, from the $\alpha^*$-well-linkedness of the set $T$ of terminals. 
Let $E'=E_H(X_i)\cup \out_H(X_i)$. Observe that $E'$ is also a subset of edges of $G$, and it disconnects $A$ from $B$ in $G$. However, $|E'|\leq 3r'<\alpha^*k/3$, a contradiction.
\end{proof}
\fi

Let $X'_i$ be the set of vertices of $G$, obtained from $X_i$ by
un-contracting the super-nodes of $X_i$. We apply Theorem~\ref{thm:
  well-linked-decomposition} to the set $X'_i$ of vertices, to obtain a
partition $\wset_i$ of $X'_i$ into $\alphaWL$-good clusters. It is easy to
see that all clusters in $\wset_i$ are acceptable, since we are
guaranteed that for each $R\in \wset_i$, $|\out(R)|\leq |\out(X'_i)|\leq
r'$, and $|R\cap T|<|T|/2$. The new partition $\cset'$ of the vertices
of $G$ into acceptable clusters is obtained as follows. We include all
clusters of $\cset$ that are disjoint from $X'_i$, and we additionally
include all clusters in $\wset_i$. From the above discussion, all
clusters in $\cset'$ are acceptable. It now only remains to bound
$\phi(\cset')$. It is easy to see that $\phi(\cset')\leq
\phi(\cset)-|E_H(X_i)|-|\out_H(X_i)|+\sum_{R\in
  \wset_i}|\out(R)|$. The choice of $\alphaWL$ ensures that $\sum_{R\in
  \wset_i}|\out(R)|\leq 1.25 |\out_G(X'_i)| = 1.25 |\out_H(X_i)|$. Since
$|E_H(X_i)|>2|\out_H(X_i)|$, we get that $\phi(\cset')\leq\phi(\cset)-1$.

From now on, we assume that for every cluster $X_i\in \xset$,
$|E(X_i)|\geq 2r'$.

\paragraph{Phase 2} 
For convenience, we assume without loss of generality, that
$\xset=\set{X_1,\ldots,X_z}$. For $1\leq i\leq z$, let
$m_i=|E(X_i)|$. Recall that from the above discussion, $m_i\geq r'$,
and $\sum_{i=1}^zm_i\geq \alpha^*k/6$. We set
$h_i=\lceil\frac{6m_ih}{\alpha^*k}\rceil$.  Let $X'_i=\bigcup_{v_C\in X_i}C$. In the
remainder of this section, we will partition, for each $1\leq i\leq
z$, the graph $G[X'_i]$ into $h_i$ vertex-disjoint subgraphs, of
treewidth at least $r$ each. Since $\sum_{i=1}^zh_i\geq
\sum_{i=1}^z\frac{6m_i h}{\alpha^* k}\geq h$, this will complete the proof of
Theorem~\ref{thm:main1-intro}.

From now on, we focus on a specific graph $H[X_i]$, and its
corresponding un-contracted graph $G[X'_i]$. Our algorithm performs
$h_i$ iterations. In the first iteration, we embed an expander over
$r''=r\poly\log k$ vertices into $H[X_i]$. We then partition $H[X_i]$
into two sub-graphs: $H_1$, containing all vertices that participate
in this embedding, and $H'_1$ containing all remaining vertices. Our
embedding will ensure that $\sum_{v\in V(H_1)}d_{H}(v)\leq
r^2\poly\log k$, or in other words, we can obtain $H'_1$ from $H[X_i]$
by removing only $r^2\poly\log k$ edges from it, and deleting isolated vertices. We then proceed to the second
iteration, and embed another expander over $r''$ vertices into
$H'_1$. This in turn partitions $H'_1$ into $H_2$, that contains the
embedding of the expander, and $H'_2$, containing the remaining
edges. In general, in iteration $j$, we start with a sub-graph $H'_{j-1}$ of $H$, and partition it into two subgraphs: $H_j$ containing the embedding of an expander, and $H_{j}'$ that becomes an input to the next iteration. Since we ensure that for each graph $H_j$, the total out-degree
of its vertices is bounded by $r^2\poly\log k$, each residual graph $H'_j$ is
guaranteed to contain a large fraction of the edges of the original
graph $H[X_i]$. We show that this in turn guarantees that $H'_j$
contains a large sub-graph with a large conductance, which will in turn
allow us to embed an expander over a subset of $r''$ vertices into
$H'_j$ in the following iteration.

We start with the following theorem that forms the technical basis for
iteratively embedding multiple expanders of certain size into a larger
expander. 

\begin{theorem}\label{thm: find large conductance subgraph}
  Let $\G$ be any graph with $|E(\G)|=m$ and $\Psi(\G)\geq \gamma$,
  where $\gamma\leq 0.1$. Let $\H$ be a sub-graph obtained from $\G$
  by removing some subset $S_0$ of vertices and all their adjacent
  edges, so $\H=\G-S_0$. Assume further that $|E(\G)\setminus E(\H)|\leq \gamma
  m/8$. Then we can efficiently compute a subset $S$ of vertices in $\H$, such that
  $\H[S]$ contains at least $m/2$ edges and has conductance at least
  $\frac{\gamma}{4\alphaSC(m)}$.
\end{theorem}

\iffull
\begin{proof}
  We show an algorithm to find the required set $S$ of vertices. We
  start with $S=V(\H)$, and a collection $\wset$ of disjoint vertex subsets
  of $\H$, that is initially empty. We then perform a number of iterations, where in each iteration we apply the algorithm $\algSC$ for approximating a minimum-conductance cut to graph $\H[S]$. If the algorithm returns a partition $(A,B)$
  of $S$, such that $|E_{\G}(A)|\leq |E_{\G}(B)|$, and $|E_{\G}(A,B)|<
  \gamma |E_{\G}(A)|/4$, then we delete the vertices of $A$ from $S$, add
  $A$ to $\wset$, and proceed to the next iteration. Otherwise, if the conductance of the computed cut is at least $\gamma/4$, then we terminate the algorithm, and we let $S^*$ denote the
  final set $S$. Then clearly, the
  conductance $\Psi(\H[S^*])\geq \frac{\gamma}{4\alphaSC(m)}$. It now only remains to
  prove that $\H[S^*]$ contains at least $m/2$ edges.

  Let $R_0$ denote the set of all edges that belong to $E(\G)$ but not
  to $E(\H)$. Notice that each such edge has at most one endpoint in
  the initial set $S$. We keep track of an edge set $R$ that is
  initialized to $R_0$. In each iteration, we add some edges to $R$,
  charging them to the edges that already belong to $R$. We will
  ensure that if $S$ is the current set, then all edges in $\out(S)$
  belong to $R$. Set $R$ will also contain all edges in $\bigcup_{A\in
    \wset}\out(A)$.

  Consider some iteration where we have computed a partition $(A,B)$
  of the current set $S$, and deleted the vertices of $A$ from $S$.
  Notice that this means that $|E_{\G}(A,B)|<\gamma
  |E_{\G}(A)|/4$. Consider the partition $(A,C)$ of the vertices of
  $\G$, where $C=B\cup (V(\G)\setminus S)$. Since $\Psi(\G)\geq
  \gamma$, $|E_{\G}(A,C)|\geq \gamma |E_{\G}(A)|$. Since
  $E_{\G}(A,C)\subseteq E_{\G}(A,B)\cup (\out_{\G}(S)\cap
  \out_{\G}(A))$, and $|E_{\G}(A,B)|<\gamma |E_{\G}(A)|/4$, we have
  that $|\out_{\G}(S)\cap \out_{\G}(A)|\geq 3\gamma|E_{\G}(A)|/4\geq
  3|E_{\G}(A,B)|$.

  The edges of $\out_{\G}(S)\cap \out_{\G}(A)$ must already belong to
  $R$. We add the edges of $E_{\G}(A,B)$ to $R$, and we charge their
  cost to the edges of $\out_{\G}(S)\cap \out_{\G}(A)$. Each edge in
  $\out_{\G}(S)\cap \out_{\G}(A)$ is then charged at most $1/3$, and
  each such edge will never be charged again, as none of its endpoints is any longer contained in the new set $S=B$.

  Using this charging scheme, the total direct charge to each edge of
  $R$ is at most $1/3$, and the total amount charged to each edge in
  $R_0$ (including direct and indirect charging) is a geometric series
  whose sum is bounded by $1$. Therefore, $|R|\leq 2|R_0|\leq \gamma
  m/4$.

  We now assume for contradiction that
  $E_{\H}(S^*)=|E_{\G}(S^*)|<m/2$. Recall that $S_0=V(\G)\setminus
  V(\H)$. Observe that for each cluster $A\in \wset$, $|E_{\G}(A)|\leq
  m/2$, since, when $A$ was added to $\wset$, there was another
  cluster $B$ disjoint from $A$ with $|E_{\G}(A)|\leq
  |E_{\G}(B)|$. Let $\wset'=\wset\cup \set{S_0,S^*}$. From the above
  discussion, for each set $Z\in \wset'$, $|E_{\G}(Z)|\leq m/2$, while
  $\sum_{Z\in \wset'}|E_{\G}(Z)|\geq m-|R|\geq 0.9m$. Using Claim~\ref{claim:
    simple partition}, we can partition the clusters in $\wset'$ into
  two subsets, $\aset,\bset$, such that $\sum_{Z\in
    \aset}|E_{\G}(Z)|,\sum_{Z\in \bset}|E_{\G}(Z)|\geq 0.3m$. Let
  $X=\bigcup_{Z\in \aset}Z$, $Y=\bigcup_{Z\in \bset}Z$. Then, since
  $\G$ has conductance at least $\gamma$, $|E_{\G}(X,Y)|\geq 0.3\gamma
  m$. However, $E_{\G}(X,Y)\subseteq R$, and as we have shown,
  $|R|\leq \gamma m/4<0.3\gamma m$, a contradiction.
\end{proof}

\fi

The next theorem is central to the execution of Phase 2. The theorem shows that, if we are given a sub-graph $H'$ of $H$ that has a high enough conductance, and contains at least $r'$ edges, then we can find a subset $S$ of $r'$ vertices of $H'$, such that the following holds: if $S'=\bigcup_{v_C\in S}C$, and $G'=G[S']$, then $\tw(G')\geq r$. In order to show this, we embed an expander over a set of $r''=r\poly\log k$ of vertices into $H'$, and define $S$ to be the set of all vertices of $H'$ participating in this embedding. The embedding of the expander into $H'[S]$ is then used to certify that the treewidth of the resulting graph $G'$ is at least $r$. The proof of the following theorem \iffull is deferred to the Appendix\fi .

\begin{theorem}\label{thm: remove expander}
  Let $H'$ be any vertex-induced subgraph of $H$, such that
  $|E(H')|\geq r'$, and $\Psi(H')\geq \frac{1}{640\alphaSC^2(m)\log m}$. Then there is an efficient algorithm to find a subset $S$ of at most $r'$ vertices of $H'$, such that, if $G'$
  is obtained from $H'[S]$ by un-contracting the super-nodes in $S$,
  then $\tw(G')\geq r$.
\end{theorem}

We are now ready to complete the description of the algorithm for Phase 2.  
Our algorithm considers each one of the subsets $X_i\in \xset$
of vertices separately. Fix some $1\leq i\leq z$.  If $h_i=1$, then by
Theorem~\ref{thm: remove expander}, graph $G[X_i]$ has treewidth at
least $r$. Otherwise, we perform $h_i$
iterations. At the beginning of every iteration $j$, we are given some
vertex-induced subgraph $H_j$ of $H[X_i]$, with $|E(H_j)|\geq m_i/2\geq r'$ and
$\Psi(H_j)\geq \frac{1}{640\alphaSC^2(m)\log m}$. At the beginning,
$H_1=H[X_i]$, and as observed before, $\Psi(H_1)\geq \frac{1}{160\alphaSC(m)\log m}\geq \frac{1}{640\alphaSC^2(m)\log m}$. In order to execute the $j$th iteration, we apply
Theorem~\ref{thm: remove expander} to graph $H'=H_j$, and compute a
subset $S$ of at most $r'$ vertices of $H'$. We denote this set of
vertices by $S_j^i$, and we let $H^i_j=H[S_j^i]$. We also let $G^i_j$
be the sub-graph of $G$, obtained by un-contracting the super-nodes of
$H^i_j$. From Theorem~\ref{thm: remove expander}, $\tw(G^i_j)\geq r$.
We then apply Theorem~\ref{thm: find large conductance subgraph} to
graph $\G=H[X_i]$, set $S_0=\bigcup_{j'=1}^jS_{j'}^i$, and $\H=\G\setminus
S_0$, to obtain the graph $H_{j+1}=\H[S]$ that becomes an input to the
next iteration.
 
In order to show that we can carry this process out for $h_i$
iterations, it is enough to prove that $\sum_{j=1}^{h_i}\sum_{v\in
  S_j^i}d_H(v)\leq \gamma m_i/8$, where $\gamma=\frac{1}{160 \alphasc(m)\log m}$. Indeed, since the vertex degrees in $H$ are bounded by $r'$,

\[\sum_{j=1}^{h_i}\sum_{v\in S_j^i}d_H(v)\leq \sum_{j=1}^{h_i}r'\cdot |S^i_j|\leq (r')^2\cdot h_i\leq O\left (\frac{m_ir^2 h\poly\log k}{\alpha^*k}\right ),\]

by substituting $h_i=\lceil \frac{6m_ih}{\alpha^*k}\rceil$. Since we assume
that $r^2h<k/\poly\log k$, and $m=O(k^6)$, it follows that the sum is
bounded by $\frac{m_i}{1280\alphasc(m)\log m}$, as required.

Our final collection of subgraphs is $\Pi=\set{G_j^i\mid 1\leq i\leq
  z,1\leq j\leq h_i}$.  From the above discussion, $\Pi$ contains
$\sum_{i=1}^zh_i\geq h$ subgraphs of treewidth at least $r$ each.

\iffull

\label{----------------------------------------sec: proof of second main thm--------------------------------------}
\section{Proof of Theorem~\ref{thm:main2-intro}}\label{sec: proof of second main theorem}
We again use Theorem~\ref{thm:log-degree} to obtain a graph $G'$ whose maximum vertex degree is $O( \log^3 k)$, and $\tw(G')=\Omega(k/\poly\log k)$. From Lemma~\ref{lem: find well-linked set}, there is an efficient algorithm to find a subset $T$ of $\Omega(k/\poly\log k)$ vertices, that we refer to as \emph{terminals} from now on, such that $T$ is $\alpha^*$-well-linked, for $\alpha^*=\Omega(1/\log k)$. For notational convenience, we denote $G'$ by $G$ and $|T|$ by $k$ from now on. Using this notation, the maximum vertex degree in $G$ is bounded by $\Delta=O(\log^3k)$.
We use a parameter $r'=2^{20}r\Delta^2 h\alphasc(k)$, and we assume that $k\geq 2^{10}h^2r'\log k/\alpha^*=2^{30}h^3r\Delta^2\log k\cdot \alphasc(k)/\alpha^*=\Omega(h^3r\poly\log k)$.

We say that a subset $C\sse V(G)$ of vertices is an \emph{acceptable} cluster, iff $|\out(C)|\leq r'$, $|C\cap T|\leq |T|/2$, and $G[C]$ is connected. As before, given a partition $\cset$ of the vertices of $G$ into acceptable clusters, we define a corresponding contracted graph $H_{\cset}$, obtained from $G$ by contracting every cluster $C\in \cset$ into a super-node $v_C$ and removing self-loops. We again denote by $\phi(\cset)$ the number of edges in the graph $H_{\cset}$. Observe that if $C$ is a cluster consisting of a single node, then $C$ is acceptable. Therefore, the partition $\set{\set{v}\mid v\in V(G)}$ of the vertices of $G$, where every vertex belongs to a separate cluster is a partition into acceptable clusters. 
The following observation is analogous to Observation~\ref{obs: many edges}. Its proof is identical and is omitted here.

\begin{observation}\label{obs: many edges 2}
Let $\cset$ be any partition of $V(G)$ into acceptable clusters, and let $H_{\cset}$ be the corresponding contracted graph. Then $|E(H_{\cset})|\geq \alpha^* k/3$.
\end{observation}

As before, our algorithm consists from a number of iterations. At the beginning of each iteration, we are given a partition $\cset$ of the vertices of $G$ into acceptable clusters, and the corresponding contracted graph $H$. The initial partition, that serves as the input to the first iteration, is $\set{\set{v}\mid v\in V(G)}$. The execution of each iteration is summarized in the following theorem.

\begin{theorem}\label{thm: iteration2}
Given a partition $\cset$ of the vertices of $G$ into acceptable clusters, there is an efficient randomized algorithm, that w.h.p. either computes a new partition $\cset'$ of the vertices of $G$ into acceptable clusters with $\phi(\cset')\leq \phi(\cset)-1$, or finds a partition of $G$ into $h$ disjoint subgraphs with treewidth at least $r$ each.
\end{theorem}

By applying theorem~\ref{thm: iteration2} to graph $G$ repeatedly, we are guaranteed to find a partition of $G$ into $h$ disjoint subsets of treewidth at least $r$ each after $O(|E(G)|)$ iterations. From now on we focus on proving Theorem~\ref{thm: iteration2}. Let $\cset$ be the current partition of $V(G)$ into acceptable clusters, and let $H$ be the corresponding contracted graph. We denote $|E(H)|$ by $m$. From Observation~\ref{obs: many edges 2}, $m\geq \alpha^*k/3$.

Our algorithm consists of two steps. In the first step, we compute a random partition of $V(H)$ into $(h+1)$ disjoint subsets $X_1,\ldots,X_{h+1}$. At most one of these subsets may contain more than half the terminals, and we ignore this subset in the second step. For each one of the remaining subsets $X_i$, we either extract a subgraph $G_i$ of $G$ whose treewidth is at least $r$, or find a new partition $\cset'$ of $V(G)$ into acceptable clusters with $\phi(\cset')\leq \phi(\cset)-1$.

We start with the first step. Partition the vertices of $V(H)$ into $(h+1)$ disjoint subsets $X_1,\ldots,X_{h+1}$, as follows. Each vertex $v\in V(H)$ chooses an index $1\leq j\leq h+1$ independently uniformly at random, and is then added to $X_j$.
The following claim is very similar to a claim that was proved in~\cite{Chuzhoy11}, and we include its proof in the Appendix for completeness, since we have changed the parameters.

\begin{claim}\label{claim: bound on random partition} With probability at least $\half$, for each $1\leq j\leq h+1$, $|\out_{H}(X_j)|< \frac{16m}{h}$, while $|E_{H}(X_j)|\geq \frac{m}{8h^2}$.
\end{claim}
\newcommand{\localalpha}{\gamma}

Given a partition $X_1,\ldots,X_{h+1}$, we can efficiently check
whether the conditions of Claim~\ref{claim: bound on random partition}
hold for it. If this is not the case, we compute a new random
partition, until the conditions of Claim~\ref{claim: bound on random
  partition} hold. Clearly, after a polynomial number of iterations,
w.h.p. we obtain a partition with desired properties.
 
From now on we assume that we are given a partition
$X_1,\ldots,X_{h+1}$ of the vertices of $H$, for which the conditions
of Claim~\ref{claim: bound on random partition} hold.  For each $1\leq
j\leq h+1$, let $X'_j$ be the subset of vertices of $G$, obtained by
un-contracting the super-nodes in $X_j$, that is,
$X'_j=\bigcup_{v_C\in X_j}C$. Notice that at most one subset $X'_j$
may contain more than $|T|/2$ terminals. We assume without loss of
generality that this subset is $X'_{h+1}$, and we ignore it from now
on.  Observe that for each $1\leq j\leq h$, $|E_G(X'_j)|\geq
|E_{H}(X_j)|>\frac{|\out_{H}(X_j)|}{128h}=\frac{|\out_G(X'_j)|}{128h}$.
In our next step, we show that for each $1\leq j\leq h$, we can either
find a subset $S_j\sse X'_j$ of $r'/\Delta$ vertices, that are
$\localalpha$-well-linked in $G[X'_j]$, for $\localalpha=\frac{6\Delta^2r}{r'}$,
or we can find a new partition $\cset'$ of $V(G)$ into acceptable
clusters with $\phi(\cset')\leq \phi(\cset)-1$.
 
 \begin{claim}\label{claim: recover large-treewidth sets}
   For each $1\leq j\leq h$, we can either find a subset $S_j\sse
   X'_j$ of $r'/\Delta$ vertices, such that $S_j$ is
   $\localalpha$-well-linked in $G[X'_j]$, for
   $\localalpha=\frac{6\Delta^2r}{r'}$, or we can compute a new partition
   $\cset'$ of $V(G)$ into acceptable clusters with $\phi(\cset')\leq
   \phi(\cset)-1$.
 \end{claim}
 
 We complete the proof of Claim~\ref{claim: recover large-treewidth
   sets} below, and show that the proof of Theorem~\ref{thm:
   iteration2} follows from it here. If we find a partition $\cset'$
 of $V(G)$ into acceptable clusters with $\phi(\cset')\leq
 \phi(\cset)-1$, then we return this partition. Otherwise, we let
 $G_j=G[X'_j]$, for $1\leq j\leq h$. From Corollary~\ref{cor: from
   well-linkedness to treewidth}, using the sets $S_j$, $\tw(G_j)\geq
 \frac{|S_j|\localalpha}{3\Delta}-1=\frac{r'}{\Delta}\cdot
 \frac{6\Delta^2r}{r'}\cdot\frac{1}{3\Delta}-1\geq r$.
 In order to complete the proof of Theorem~\ref{thm:main2-intro}, it
 is therefore enough to prove Claim~\ref{claim: recover
   large-treewidth sets}.

 \begin{proof}
   Fix some $1\leq j\leq h$. We maintain a partition $\wset_j$ of the
   vertices of $X'_j$, where at the beginning, $\wset_j=X'_j$. We then
   perform a number of iterations, as follows.
 
   In every iteration, we select any cluster $C\in \wset_j$ with
   $|\out(C)|\geq r'$.  Let $\Gamma$ be any subset of $r'$ edges in
   $\out(C)$. We set up an instance of the sparsest cut problem, as
   follows. Subdivide every edge $e\in \Gamma$ by a vertex $t_e$, and
   let $T'=\set{t_e\mid e\in \Gamma}$. Consider the sub-graph of the
   resulting graph induced by $C\cup T'$, where the vertices of $T'$
   serve as terminals. We apply the algorithm $\algsc$ to the
   resulting instance of the sparset cut problem. Let $(A,B)$ be the
   resulting partition of $C$, and let $(\Gamma_A,\Gamma_B)$ be the
   resulting partition of the edges of $\Gamma$, that is,
   $\Gamma_A=\out(A)\cap \Gamma$, and $\Gamma_B=\out(B)\cap
   \Gamma$. Assume without loss of generality that $|\out(A)|\leq
   |\out(B)|$. Two cases are possible. If $|E(A,B)|\geq \localalpha\cdot
   \alphasc(r') \min\set{|\Gamma_A|,|\Gamma_B|}$, then we define the
   set $S_j$ to contain the endpoints of the edges of $\Gamma$ that
   belong to $C$. Since the degree of every vertex in $G$ is at most
   $\Delta$, $|S_j|\geq r'/\Delta$, and since the algorithm $\algsc$
   returned a cut whose sparsity is at least $\localalpha\cdot
   \alphasc(r')$, it follows that set $S_j$ is $\localalpha$-well-linked in
   $C$, and hence in $G[X'_j]$.
 
   Otherwise, $|E(A,B)|< \localalpha\cdot \alphasc(r') |\Gamma_A|$. Every
   edge in $\out(A)\cap \out(C)$ is then charged
   $|E(A,B)|/|\out(A)\cap \out(C)|\leq
   r'\localalpha\alphasc(r')/|\out(A)\cap \out(C)|$ for the edges of
   $E(A,B)$. Notice that the total charge is at least $|E(A,B)|$. The
   charge to every edge of $\out(A)\cap \out(C)$ can also be bounded
   by $\localalpha\cdot \alphasc(r')$, since $|E(A,B)|<\localalpha\cdot
   \alphasc(r')|\out(A)\cap \out(C)|$.
 
   The algorithm terminates when we either find the desired subset
   $S_j$ of vertices, or when every cluster $C\in \wset_j$ has
   $|\out(C)|<r'$. In the former case, we return the set $S_j$ as the
   output. In the latter case, we build a partition $\cset'$ of $V(G)$
   into acceptable clusters, with $\phi(\cset')\leq\phi(\cset)-1$.
 The collection $\cset'$ of clusters contains all clusters $C\in
 \cset$ with $C\cap X'_j=\emptyset$. Additionally, for each cluster
 $C\in \wset_j$, we add all connected components of $G[C]$ to
 $\cset'$. It is easy to see that $\cset'$ is an acceptable
 clustering. It now only remains to show that $\phi(\cset')\leq
 \phi(\cset)-1$. Observe that:
 
 \[\phi(\cset')\leq\phi(\cset)-|\out_G(X'_j)|-|E_G(X'_j)|+\sum_{C\in \wset_j}|\out_G(C)|\]
 
 We show that $\sum_{C\in
   \wset_j}|\out_G(C)|<\left(1+\frac{1}{2^8h}\right
 )|\out_G(X'_j)|$. Since $|E_G(X'_j)|\geq
 \frac{|\out_G(X'_j)|}{128h}$, we will obtain that
 $\phi(\cset')\leq\phi(\cset)-1$.
 
 In order to bound $\sum_{C\in \wset_j}|\out(C)|$, we use the charging
 scheme defined above. Consider some cluster $C$ that belonged to
 $\wset_j$ at some point of the algorithm execution, and assume that
 we have replaced $C$ with $A$ and $B$, where $|\out(A)|\leq
 |\out(B)|$, charging every edge in $\out(A)\cap \out(C)$ at most
 $r'\localalpha\alphasc(r')/|\out(A)\cap \out(C)|$. Recall that
 $|E(A,B)|\leq \localalpha\alphasc(r')\cdot |\Gamma_A|\leq
 \localalpha\alphasc(r')\cdot |\out(A)\cap \out(C)|<0.1|\out(A)\cap
 \out(C)|$, since $\localalpha=\frac{6\Delta^2r}{r'}$, and
 $r'=2^{20}r\Delta^2 h\alphasc(k)$.  Therefore,
 $|\out(A)|<2|\out(C)|/3$.  The charge to the edges of $\out(A)\cap
 \out(C)$ can be bounded by $\frac{r'\localalpha\alphasc(r')}{|\out(C)\cap
   \out(A)|}\leq \frac{2r'\localalpha\alphasc(r')}{|\out(A)|}$, since
 $|\out(A)|=|\out(C)\cap \out(A)|+|E(A,B)|<2|\out(C)\cap \out(A)|$.
 
 Consider now some edge $e=(u,v)$. We bound the charge to the edge $e$
 via the vertex $u$. Let $C_1,C_2,\ldots,C_z$ be the clusters that
 belonged to $\wset_j$ over the course of the algorithm, such that for
 each $1\leq i\leq z$, $u\in C_i$, $v\not\in C_i$, and $e$ was charged
 via $u$ when cluster $C_i$ was created. Then for each $2\leq i\leq
 z$, $|\out(C_i)|<2|\out(C_{i-1})|/3$ must hold, and edge $e$ was
 charged at most $\frac{2r'\localalpha\alphasc(r')}{|\out(C_{i})|}$ for the
 creation of cluster $C_i$. Moreover, we are guaranteed that
 $|\out(C_{z-1})|> r'$.  Therefore, the total charge to $e$ via $u$
 for creating clusters $C_1,\ldots,C_{z-1}$ is bounded by:

 \[\frac{2r'\localalpha\alphasc(r')}{|\out(C_1)|}+\frac{2r'\localalpha\alphasc(r')}{|\out(C_2)|}+\cdots
 +\frac{2r'\localalpha\alphasc(r')}{|\out(C_{z-1})|}\leq
 \frac{2r'\localalpha\alphasc(r')}{r'}\left (1+\frac{2}3+\left (\frac 2
     3\right )^2+\cdots \right )\leq 6\localalpha\alphasc(r').\]

 In the last iteration, $e$ is charged at most $\localalpha\alphasc(r')$
 for creating cluster $C_z$. Therefore, the total direct charge to $e$
 via $u$ is bounded by $7\localalpha\alphasc(r') <\frac{1}{2^{11}h}$, since
 $\localalpha=\frac{6\Delta^2r}{r'}$ and $r'=2^{20}r\Delta^2 h\alphasc(k)$.
 The total direct charge to edge $e$, via both $u$ and $v$, is then
 bounded by $\frac{1}{2^{10}h}$, and the total charge to any edge
 $e\in \out(X'_j)$, including the direct and the indirect charge (that
 happens when $e$ is charged for some edge $e'$, which is in turn
 charged for some other edges), is at most $\frac{1}{2^9h}$, since the
 indirect charge forms a geometrically decreasing sequence. We
 conclude that $\sum_{C\in
   \wset_j}|\out(C)|<\left(1+\frac{1}{2^8h}\right )|\out_G(X'_j)|$,
 and $\phi(\cset')\leq \phi(\cset)-1$, as required.
 \end{proof}

\fi

\label{----------------------------------------sec: apps------------------------}
\section{Applications}
\label{sec:apps}

We now describe two applications of Theorem~\ref{thm:main1-intro}.
Consider an integer-valued {\em parameter} $P$ that associates a
number $P(G)$ with each graph $G$. For instance, $P(G)$ could be the
size of the smallest vertex cover of $G$, or it could be the maximum
number of vertex-disjoint cycles in $G$. We say that $P$ is
\emph{minor-closed} if $P(G) \ge P(H)$ for any minor $H$ of $G$, that
is, the value does not increase when deleting edges or contracting
edges. A number of interesting parameters are minor-closed. Following
\cite{DemaineH-survey}, we say that $P$ has the
\emph{parameter-treewidth bound}, if there is some function
$f:\mathbb{Z}_+ \rightarrow \mathbb{Z}_+$ such that $P(G) \le k$
implies that $\tw(G) \le f(k)$. In other words, if the treewidth of
$G$ is large, then $P(G)$ must also be large. A minor-closed property
$P$ has the parameter-treewidth bound iff it has the bound on the
family of grids. This is because an $r \times r$ grid has treewidth
$r$, and the \GM Theorem shows that sufficiently large treewidth
implies the existence of a large grid minor.  This approach also has
the advantage that grids are simple and concrete graphs to reason
about. However, this approach for proving parameter treewidth bounds
suffers from the (current) qunatitative weakness in the \GM theorem.
For a given parameter $P$, one can of course 
focus on methods that are tailored to it. Alternatively,
good results can be obtained in special classes of graphs such as
planar graphs, and graphs that exclude a fixed graph as a minor, due
to the linear relationship between the treewidth and the grid-minor
size in such graphs.  Theorem~\ref{thm:main1-intro} allows for a
generic method to change the dependence $f(k)$ from exponential to
polynomial, under some mild restrictions. The following subsections
describe these applications.

\subsection{FPT Algorithms in General Graphs}
\label{subsec:fpt}

Let $P$ be any minor-closed graph parameter, and consider the decision
problem associated with $P$: Given a graph $G$ and an integer $k$, is
$P(G) \le k$? We say that parameter $P$ is \emph{fixed-parameter
  tractable}, iff there is an algorithm for this decision problem,
whose running time is $h(k) \cdot n^{O(1)}$ where $n$ is the size of
$G$ and $h$ is a function that depends only on $k$. There is a vast
literature on \FPT (FPT), and we refer the reader to
\cite{DowneyF07,Niedermeier2006,FlumG10,Bodlaenderetal2012}.

Observe that for any minor-closed parameter $P$, and any fixed integer
$k$, the family $\fset=\set{G\mid P(G)\leq k}$ of graphs is a
minor-closed family. That is, if $G\in \fset$, and $G'$ is a minor of
$G$, then $G'\in \fset$. Therefore, from the work of Robertson and
Seymour on graph minors and the proof of Wagner's conjecture, there is
a finite family $H_\fset$ of graphs, such that $\fset$ is precisely
the set of all graphs that do not contain any graph from $H_\fset$ as
a minor.
In particular, in order to test whether $P(G)\leq k$, we only need to
check whether $G$ contains a graph from $H_\fset$ as a minor, and this
can be done in time $O(n^3)$ (where we assume that $k$ is a constant),
using the work of Robertson and Seymour. However, even though the
family $H_\fset$ of graphs is known to exist, no explicit algorithms
for constructing it are known. The family $H_{\fset}$ of course
depends on $P$, and moreover, even for a fixed property $P$, it varies
with the parameter $k$. Therefore, the theory only guarantees the
existence of a non-uniform FPT algorithm for every minor-closed
parameter $P$. For this reason, it is natural to consider various
restricted classes of minor-closed parameters. Motivated by the
existence of sub-exponential time algorithms on planar and
H-minor-free graphs, a substantial line of work has focused on {\em
  bidimensional} parameters --- see Demaine et
al. \cite{Demaineetal-bidimensional}, and the survey in
\cite{DemaineH-survey}.  Demaine and Hajiaghayi \cite{DemaineH07}
proved the following generic theorem on \FPT of minor-closed
bidimensional properties that satisfy some mild additional conditions.

\begin{theorem}[\cite{DemaineH07}]
  \label{thm:demaineh-fpt}
  Consider a minor-closed parameter $P$ that is positive on some $g
  \times g$ grid, is at least the sum over the connected components of
  a disconnected graph, and can be computed in $h(w) n^{O(1)}$ time
  given a width-$w$ tree decomposition of the graph. Then there is an
  algorithm that decides whether $P$ is at most $k$ on a graph with
  $n$ vertices in $[2^{2^{O(g\sqrt{k})^5}} +
  h(2^{O(g\sqrt{k})^5})]n^{O(1)}$ time.
\end{theorem}

The main advantage of the above theorem is its generality.
However, its proof uses the \GM Theorem, and hence the running
time of the algorithm is doubly exponential in the parameter $k$.
Demaine and Hajiaghayi also observed, in the following theorem,
that the running time can be reduced to singly-exponential in 
$k$ if the \GM Theorem can be improved substantially.

\begin{theorem}[\cite{DemaineH07}]
  Assume that every graph of treewidth greater than $\Theta(g^2 \log
  g)$ has a $g \times g$ grid as a minor. Then for every minor-closed
  parameter $P$ satisfying the conditions of Theorem
  \ref{thm:demaineh-fpt}, there is an algorithm that decides whether
  $P(G)\leq k$  on any $n$-vertex graph $G$ in $[2^{O(g^2k\log
    (gk))} +h(O(g^2k \log (gk)))]n^{O(1)}$ time.
\end{theorem}

We show below that, via Theorem~\ref{thm:main1-intro}, we can 
bypass the need to improve the \GM Theorem.

\begin{theorem}
  \label{thm:fpt}
  Consider a minor-closed parameter $P$ that is positive on all graphs
  with treewidth $\ge p$, is at least the sum over the connected
  components of a disconnected graph, and can be computed in $h(w)
  n^{O(1)}$ time given a width-$w$ tree decomposition of the
  graph. Then there is an algorithm that decides whether $P$ is at
  most $k$ on a graph with $n$ vertices in $[2^{\tilde{O}(p^2 k)} +
  h(\tilde{O}(p^2 k))] n^{O(1)}$ time.
\end{theorem}

\begin{proof}
  Let $k' = \tilde{\Theta}(p^2 k)$. If the given graph $G$ has
  treewidth greater than $k'$, then by Theorem~\ref{thm:main1-intro}
  it can be partitioned into $k$ node-disjoint subgraphs
  $G_1,\ldots,G_k$ where $\tw(G_i) \ge p$ for each $i$. Let $G'$ be
  obtained by the union of these disconnected graphs (equivalently we
  remove the edges that do not participate in the graphs $G_i$ from
  $G$).  From the assumptions on $P$, $P(G_i) \ge 1$ for each $i$, and
  $P(G') \ge \sum_i P(G_i) \ge k$. Moreover, since $P$ is
  minor-closed, $P(G) \ge P(G')$. Therefore, if $\tw(G) \ge k' =
  \tilde{\Omega}(p^2 k)$ then $P(G)\geq k$ must hold.
  
  We use known algorithms, for instance \cite{Amir10}, that, given a
  graph $G$, either produce a tree decomposition of width at most $4w$
  or certify that $\tw(G) > w$ in $2^{O(w)} n^{O(1)}$ time. Using such
  an algorithm we can detect in $2^{O(k')}n^{O(1)}$ time whether $G$
  has treewidth at least $k'$, or find a tree decomposition of width
  at most $4k'$.

  If $\tw(G) \ge k'$, then, as we have argued above, $P(G)\geq k$. We
  then terminate the algorithm with a positive answer. Otherwise,
  $\tw(G) < 4k'$ and we can use the promised algorithm that runs in
  time $h(4k') \cdot n^{O(1)}$ to decide whether $P(G) < k$ or not.
  The overall running time of the algorithm is easily seen to be the
  claimed bound.
\end{proof}

\begin{remark}
  In the proof of Theorem~\ref{thm:fpt} the assumption on $P$ being
  minor-closed is used only in arguing that $P(G') \ge P(G)$. Thus, it
  suffices to assume that the parameter $P$ does not increase under
  edge deletions (in addition to the assumption on $P$ over 
  disconnected components of a graph).
\end{remark}

Note that the running time is singly-exponential in $p$ and $k$.  How
does one prove an upper bound on $p$, the minimum treewidth guaranteed
to ensure that the parameter value is positive? For some problems it
may be easy to directly obtain a good bound on $p$. The following
corollary shows that one can always use grid minors to obtain a bound
on $p$.  The run-time dependence on the grid size $g$ is doubly
exponential since we are using the \GM Theorem, but it is only
singly-exponential in the parameter $k$. Thus, if $g$ is considered to
be a fixed constant, we obtain singly-exponential \FPT algorithms in
general graphs for all the problems that satisfy the conditions in
Theorem~\ref{thm:demaineh-fpt}.

\begin{corollary}
  \label{cor:fpt}
  Consider a minor-closed parameter $P$ that is positive on some $g
  \times g$ grid, is at least the sum over the connected components of
  a disconnected graph, and can be computed in $h(w) n^{O(1)}$ time
  given a width-$w$ tree decomposition of the graph. Then there is an
  algorithm that decides whether $P(G)\leq k$ on a graph $G$ with $n$
  vertices in $[2^{\tilde{O}(k\cdot 2^{O(g^{5})})} +
  h(\tilde{O}(2^{O(g^{5})} k))] n^{O(1)}$ time.
\end{corollary}

\begin{proof}
  If $P$ is minor-closed and positive on some $g \times g$ grid then
  it is positive on a graph $G$ with treewidth $p > 20^{2g^5}$
  via Theorem~\ref{thm:RST-grid}. Plugging this value of $p$ into
  the bound from Theorem~\ref{thm:fpt} gives the desired result.
\end{proof}

\begin{remark}
The results in \cite{ReedW-grid,KreutzerT10} can be used
to obtain a singly exponential dependence on $g$, 
provided $P$ can be shown to be positive on a graph that
contains a {\em grid-like} minor of size $g$.
\end{remark}

\subsection{Bounds for  \Erdos-\Posa theorems}
\label{subsec:erdosposa}

Let $\mF$ be any family of graphs. Following the notation in
\cite{Reed-chapter}, we say that the $\mF$-packing number of $G$,
denoted by $p_\mF(G)$, is the maximum number of node-disjoint
subgraphs of $G$, each of which is isomorphic to a member of $\mF$. An
$\mF$-cover is a set $X$ of vertices, such that $p_\mF(G-X) = 0$;
that is, removing $X$ ensures that there is no subgraph isomorphic to
a member of $\mF$ in $G$. The $\mF$-covering number of $G$, denoted by
$c_\mF(G)$ is the minimum cardinality of an $\mF$-cover for $G$.  It
is clear that $p_\mF(G) \le c_\mF(G)$ always holds. A family $\mF$ is said to
satisfy the \EP property if there is function $f: \mathbb{Z}_+
\rightarrow \mathbb{Z}_+$ such that $c_\mF(G) \le f(p_\mF(G))$ for all graphs $G$. 
\iffull
In other words, for every integer $k$, either $G$ has $k$ disjoint copies
of graphs from $\mF$, or there is a set of $f(k)$ nodes, whose removal
from $G$ ensures that no subgraph isomorphic to a graph from $\mF$ remains in $G$. 
\fi
\Erdos and \Posa
\cite{ErdosP65} showed such a property when $\mF$ is the family of
cycles, with $f(k) = \Theta(k \log k)$.

There is an important connection between treewidth and \EP property
as captured by the following two lemmas.
\iffull The proofs closely follow the arguments in~\cite{Thomassen88,FominST11}
 and appear in the Appendix.
\fi 

\begin{lemma}
  \label{lem:thomassen}
  Let $\mF$ be any family of connected graphs, and let $h_{\fset}$ be
  an integer-valued function, such that the following holds.  For any
  integer $k$, and any graph $G$ with $\tw(G)\geq h_{\fset}(k)$, $G$
  contains $k$ disjoint subgraphs $G_1,\ldots,G_k$, each of which is
  isomorphic to a member of $\fset$. Then $\fset$ has the \EP property
  with $f_\mF(k) \le k\cdot h_\mF(k)$.
\end{lemma}


\begin{lemma}
  \label{lem:ep-improved}
  Let $\mF$ be any family of connected graphs, and let $h_{\fset}$ be
  an integer-valued function, such that the following holds.  For any
  integer $k$, and any graph $G$ with $\tw(G)\geq h_{\fset}(k)$, $G$
  contains $k$ disjoint subgraphs $G_1,\ldots,G_k$, each of which is
  isomorphic to a member of $\fset$. Moreover, suppose that
  $h_\fset(\cdot)$ is superadditive\footnote{We say that an integer-valued function $h$ is superadditive if for all $x,y\in\mathbb{Z}^+$, $h(x)+h(y)\leq h(x+y)$.} and satisfies the property that
  $h_\fset(k+1) \le \alpha h_\fset(k)$ for all $k \ge 1$ where $\alpha$ is
  some universal constant. Then $\fset$ has the \EP property with
  $f_\mF(k) \le \beta h_\mF(k) \log (k+1)$ where $\beta$ is a universal
  constant.
\end{lemma}

One way to prove that $p_\mF(G) \ge k$ whenever $\tw(G) \ge h_\mF(k)$
is via the following proposition, that is based on the \GM Theorem.
It is often implicitly used; see \cite{Reed-chapter}.

\begin{proposition}
    \label{prop:erdos-posa-grid}
    Let $\mF$ be any family of connected graphs, and assume that there
    is an integer $g$, such that any graph containing a $g\times g$
    grid as a minor is guaranteed to contain a sub-graph isomorphic to
    a member of $\mF$. Let $h(g')$ be the treewidth that guarantees the
    existence of a $g'\times g'$ grid minor in any graph.  Then
    $f_\mF(k) \le O(k\cdot h(g\sqrt k))$. In particular $f_\mF(k) \le
    2^{O(g^5k^{2.5})}$.
\end{proposition}

We improve the exponential dependence on $k$ in the preceding
proposition to near-linear. We state a more general
theorem\ and
then derive the improvement as a corollary.

\begin{theorem}
  \label{thm:erdos-posa}
  Let $\mF$ be any family of connected graphs, and assume that there
  is an integer $r$, such that any graph of treewidth at least $r$ is
  guaranteed to contain a sub-graph isomorphic to a member of $\mF$.
  Then $f_\mF(k) \le \tilde{O}(k r^2)$.
\end{theorem}

\iffull
\begin{proof}
  Let $G$ be any graph with $\tw(G) \ge k r^2
  \polylog(kr)$. Theorem~\ref{thm:main1-intro} guarantees that $G$ can
  be partitioned into $k$ node-disjoint subgraphs $G_1,\ldots,G_k$,
  such that for each $i$, $\tw(G_i) \ge r$. From the assumption in the
  theorem statement, each $G_i$ has a subgraph isomorphic to a member
  of $\mF$.  Therefore $G$ contains $k$ such subgraphs. We have thus
  established that if $\tw(G) \ge
 \tilde{O}(k r^2)$, then $p_\mF(G) \ge k$. We apply
  Lemma~\ref{lem:ep-improved} to conclude that $f_\mF(G) \le
  \tilde{O}(k r^2)$.
\end{proof}

Combining the preceding theorem with the \GM Theorem gives the
following easy corollary.
\fi

\begin{corollary}
  \label{cor2:erdos-posa}
  Let $\mF$ be any family of connected graphs, such that for some integer $g$, any graph containing a $g\times g$ grid
  as a minor is guaranteed to contain a sub-graph isomorphic to a
  member of $\mF$. Then $f_\mF(k) \le 2^{O(g^5)} \tilde{O}(k)$.
\end{corollary}

{\bf Some concrete results:} For a fixed graph $H$, let $\mF(H)$ be
the family of all graphs that contain $H$ as a minor. Robertson and
Seymour \cite{RS-grid}, as one of the applications of their \GM
Theorem, showed that $\mF(H)$ has the \EP property iff $H$ is
planar. The if direction can be deduced as follows. Every planar graph
$H$ is a minor of a $g\times g$ grid, where $g=O(|V(H)|^2)$. We can
then use Proposition~\ref{prop:erdos-posa-grid} to obtain a bound on
$f_{\fset(H)}$, which is super-exponential in $k$. However, by
directly applying Corollary~\ref{cor2:erdos-posa}, we get the
following improved near-linear dependence on $k$.

\begin{theorem}
  \label{thm:erdos-posa-planar}
  For any fixed planar graph $H$, the family $\fset(H)$ of graphs has
  the \EP property with $f_{\mF(H)}(k) = O(k \cdot \polylog(k))$.
\end{theorem}

For any integer $m>0$, let $\fset_{m}$ be the family of all cycles
whose length is $0$ modulo $m$.  Thomassen \cite{Thomassen88} showed
that $\mF_{m}$ has the \EP property, with
$f_{\fset_{m}}=2^{m^{O(k)}}$.  We can use
Corollary~\ref{cor2:erdos-posa} to obtain a bound of
$f_{\fset_{m}}=\tilde O(k)\cdot 2^{\poly(m)}$, using the fact that a
graph containing a grid minor of size $2^{\poly(m)}$ must contain a
cycle of length $0$ modulo $m$.

\medskip
\noindent {\bf Acknowledgements:} CC thanks Tasos Sidiropoulos for
discussions on grid minors and bidemsionality, and Alina Ene for
discussions on the decomposition result in \cite{Chuzhoy11} and for
pointing out \cite{FominST11}.  This
paper is partly inspired by the lovely article of Bruce Reed on
treewdith and their applications \cite{Reed-chapter}. We thank Paul
Wollan for pointers to recent work on the \GM theorem and \EP type
theorems.

\iffull
\bibliographystyle{alpha}
\fi
\bibliography{treewidth-decomposition}

\iffull
\label{------------------------------------Appendix-----------------------------------}
\appendix

\section{Proofs Omitted from Section~\ref{sec:prelim}}
\label{--------------------------------------sec proof from prelims-------------------------}
\label{sec: proofs from prelims}

\subsection{Proof of Claim~\ref{claim: simple partition}}

  We assume without loss of generality that $x_1\geq x_2\geq\cdots\geq x_n$, and process
  the integers in this order. When $x_i$ is processed, we add $i$ to
  $A$ if $\sum_{j\in A}x_j\leq \sum_{j\in B}x_j$, and we add it to $B$
  otherwise. We claim that at the end of this process, $\sum_{i\in
    A}x_i,\sum_{i\in B}x_i\geq N/3$. Indeed, if $x_1\geq N/3$, then
  $1$ is added to $A$, and, since $x_1\leq 2N/3$, it is easy to see
  that both subsets of integers sum up to at least $N/3$.  Otherwise,
  $|\sum_{i\in A}x_i-\sum_{i\in B}x_i|\leq \max_i\set{x_i}\leq x_1\leq
  N/3$.

\subsection{Proof of Corollary~\ref{cor: from well-linkedness to treewidth}}

  We build a bramble $\bset$ of order at least $\frac{\alpha \cdot
    |T|}{3\Delta}$, as follows. Let $X$ be any subset of fewer than
  $\frac{\alpha\cdot |T|}{3\Delta}$ vertices, and let $\cset$ be the
  set of all connected components of $G\setminus X$. We claim that at
  least one connected component in $\cset$ must contain more than
  $\half |T|$ of the vertices of $T$.

  Assume otherwise. Let $\cset=\set{C_1,\ldots,C_{\ell}}$; let
  $C_{\ell+1}=X$, and let $\cset'=\cset\cup\set{C_{\ell+1}}$. Then,
  using Claim~\ref{claim: simple partition}, we can find a partition
  $(\aset,\bset)$ of the sets in $\cset'$, with $\sum_{C\in
    \aset}|C\cap T|, \sum_{C\in \bset}|C\cap T|\geq |T|/3$. Let
  $A=\bigcup_{C\in \aset}C$, $B=\bigcup_{C\in \bset}C$. From the
  well-linkedness of the set $T$ of vertices, $|E(A,B)|\geq \alpha
  |T|/3$ must hold. However, $E(A,B)$ only contains edges incident to
  the vertices of $X$, and their number is bounded by $|X|\cdot
  \Delta<\alpha |T|/3$, a contradiction.

  We are now ready to define the bramble $\bset$. For each subset $X$
  of fewer than $\frac{\alpha\cdot |T|}{3\Delta}$ vertices, we add the
  unique connected component $C_X$ of $G\setminus X$, containing more
  than half the vertices of $T$ to $\bset$. It is easy to see that
  $\bset=\set{C_X\mid X\sse V(G); |X|< \frac{\alpha \cdot
      |T|}{3\Delta}}$ is indeed a bramble. The order of $\bset$ is at
  least $\frac{\alpha |T|}{3\Delta}$, since for each set $S$ of fewer
  than $\frac{\alpha |T|}{3\Delta}$ vertices, there is a graph in
  $\bset$ that does not contain any vertices of $S$ - the graph $C_S$.
  Therefore, $\BN(G)\geq \frac{\alpha |T|}{3\Delta}$, and so
  $\tw(G)\geq \frac{\alpha |T|}{3\Delta}$-1.

\subsection{Proof of Theorem~\ref{thm: well-linked-decomposition}}

We use the $\alphasc(k')$-approximation algorithm $\algsc$ for the
sparsest cut
problem; if a polynomial-time algorithm is not needed we can use
an exact algorithm for the sparsest cut problem. 
Throughout the algorithm, we maintain a partition $\wset$ of the input
set $S$ of vertices, where for each $R\in \wset$, $|\out(R)|\leq
|\out(S)|$. At the beginning, $\wset$ consists of the subsets of $S$
defined by the connected components of $G[S]$.

Let $R\in \wset$ be any set in the current partition, and let
$(G_R,\tset'_R)$ be the instance of the sparsest cut problem
corresponding to $R$, defined as follows. We start with graph $G$, and
subdivide every edge $e\in \out_G(R)$ with a new vertex $t_e$, letting
$\tset'_R=\set{t_e\mid e\in \out_G(R)}$. Let $G_R$ be the sub-graph of
the resulting graph, induced by $R\cup \tset'_R$. We then consider the
instance of the sparsest cut on graph $G_R$, with the set $\tset'_R$
of terminals.  We say that a cut $(A',B')$ in $G_R$ is sparse, iff its
sparsity is less than $\alpha\cdot \alphasc(k')$. We apply the
algorithm $\algsc$ to the instance $(G_R,\tset'_R)$ of sparsest
cut. If the algorithm returns a cut $(A',B')$, that is a sparse cut,
then let $A=A'\setminus \tset'_R$, and $B=B'\setminus \tset'_R$. We
remove $R$ from $\wset$, and add $A$ and $B$ to it instead. Let
$T_A=\out(R)\cap \out(A)$, and $T_B=\out(R)\cap \out(B)$, and assume
without loss of generality that $|T_A|\leq |T_B|$. Then $|E(A,B)|< \alpha\cdot
\alphasc(k') |T_A|$ must hold, and in particular, $|\out(A)|\leq
|\out(B)|\leq |\out(R)|\leq |\out(S)|$. For accounting purposes, each
edge in set $T_A$ is charged $\alpha\cdot \alphasc(k')$ for the edges
in $E(A,B)$. Notice that the total charge to the edges in $T_A$ is $
\alpha \cdot \alphasc(k')|T_A|\geq |E(A,B)|$. Notice also that since
$|T_A|\leq |\out(R)|/2$ and $|E(A,B)|\leq \alpha\cdot \alphasc(k')
|T_A|\leq 0.1|T_A|$, we have that $|\out(A)|\leq 0.51|\out(R)|$.

The algorithm stops when for each set $R\in \wset$, the procedure
$\algsc$ returns a cut that is not sparse. We argue that this means
that each set $R\in \wset$ is $\alpha$-good. Assume otherwise, and let
$R\in \wset$ be a set that is not $\alpha$-good. Then the
corresponding instance of the sparsest cut problem must have a cut of
sparsity less than $\alpha$. The algorithm $\algsc$ should then have
returned a cut whose sparsity is less than $\alpha\cdot\alphasc(k')$,
that is a sparse cut.

Finally, we need to bound $\sum_{R\in\wset}|\out(R)|$. We use the
charging scheme defined above. Consider some iteration where we
partition the set $R$ into two subsets $A$ and $B$, with $|T_A|\leq
|T_B|$. Recall that each edge in $T_A$ is charged $\alpha\cdot
\alphasc(k')$ in this iteration, while $|\out(A)|\leq 0.51|\out(R)|$
holds.  Consider some edge $e=(u,v)$. Whenever $e$ is charged via the
vertex $u$, the size of the set $\out(R)$, where $u\in R\in \wset$
goes down by the factor of at least $0.51$. Therefore, $e$ can be
charged at most $2\log k'$ times via each of its endpoints. The total
charge to $e$ is then at most $4\alpha\cdot \alphasc(k')\log k'<\half$
(since $\alpha<\frac{1}{8\alphasc(k')\cdot \log k'}$). This however
only accounts for the \emph{direct} charge. For example, some edge
$e'\not \in \out(S)$, that was first charged to the edges in
$\out(S)$, can in turn be charged for some other edges. We call such
charging \emph{indirect}. If we sum up the indirect charge for every
edge $e\in \out(S)$, we obtain a geometric series, and so the total
direct and indirect amount charged to every edge $e\in \out(S)$ is at
most $8\alpha\cdot \alphasc(k')\log k'$. Therefore, $\sum_{R\in
  \wset}|\out(R)|\leq k'(1+16\alpha\cdot \alphasc(k')\log k')$ (we
need to count each edge $e\in \left (\bigcup_{R\in \wset}\out(R)\right
)\setminus \out(S)$ twice: once for each its endpoint).\hfill

\subsection{Proof of Theorem~\ref{thm:log-degree}}

Since $G$ has treewidth $k$, we can efficiently find a set $X$ of $\Omega(k)$ vertices of $G$ with properties guaranteed by Lemma~\ref{lem: find well-linked set}. Assume for
simplicity that $|X|$ is even.  We use the cut-matching game of
Khandekar, Rao and Vazirani \cite{KRV}, defined as follows.
 
We are given a set $V$ of nodes, where $\card{V}$ is even, and two
players, the cut player and the matching player. The goal of the cut
player is to construct an edge-expander in as few iterations as
possible, whereas the goal of the matching player is to prevent the
construction of the edge-expander for as long as possible. The two
players start with a graph $\mX$ with node set $V$ and an empty edge
set. The game then proceeds in iterations, each of which adds a set of
edges to $\mX$.  In iteration $j$, the cut player chooses a partition
$(Y_j, Z_j)$ of $V$ such that $\card{Y_j} = \card{Z_j}$ and the
matching player chooses a perfect matching $M_j$ that matches the
nodes of $Y_j$ to the nodes of $Z_j$.  The edges of $M_j$ are then
added to $\mX$.  Khandekar, Rao, and Vazirani \cite{KRV} showed that
there is a strategy for the cut player that guarantees that after
$O(\log^2{\card{V}})$ iterations the graph $\mX$ is a
$1/2$-edge-expander. Orecchia \etal \cite{OrecchiaSVV08} strengthened
this result by showing that after $O(\log^2{\card{V}})$ iterations the
graph $\mX$ is a $\Omega(\log{\card{V}})$-edge-expander.  
We use $\cKRV(n)$ to denote the number of
iterations of the cut-matching game required in the proof of the
preceding theorem for $|V|=n$. Note that the resulting expander is 
regular with vertex degrees equal to $\cKRV(n)$.

Using the cut-matching game we can embed an expander $H=(X,F)$ into
$G$ as follows. Each iteration $j$ of the cut-matching game requires
the matching player to find a matching $M_j$ between a given partition
of $X$ into two equal-sized sets $Y_j,Z_j$. From Lemma~\ref{lem: find well-linked set}, there exist a collection $P_j$ of
paths from $Y_j$ to $Z_j$, that cause congestion at most $1/\alpha^*$ on the vertices of $G$; these paths naturally define
the required matching $M_j$. The game terminates in $\cKRV(|X|)$
steps. Consider the collection of paths $\mP = \cup_j P_j$ and let
$G'$ be the subgraph of $G$ induced by the union of the edges in these
paths and let $H=(X,F)$ be the expander on $X$ created by the union of
the edges in $\cup_j M_j$. By the construction, for each $j$, any node $v$
of $G$ appears in at most $1/\alpha^*$ paths in $P_j$.  Therefore, the maximum
degree in $G'$ is at most $2 \cKRV(|X|)/\alpha^*=O(\log^3k)$ and moreover the node (and
hence also edge) congestion caused by the edges of $H$ in $G$ is also
upper bounded by the same quantity. We claim that $\tw(G') =
\Omega(k/\log^6 k)$. Since $H=(X,F)$ is an edge-expander, $X$ is
$\alpha$-edge-well-linked in $H$ for a fixed constant $\alpha$. Since
$H$ is embedded in $G'$ with congestion at most $2\cKRV(|X|)/\alpha^*$, $X$ is
$\frac{\alpha\cdot \alpha^*}{2\cKRV(|X|)}$-edge-well-linked in $G'$. Since the maximum
degree in $G'$ is at most $2\cKRV(|X|)/\alpha^*$, we can apply
Corollary~\ref{cor: from well-linkedness to treewidth} to see that
$\tw(G') = \Omega\left(\frac{|X|(\alpha^*)^2}{(\cKRV(|X|))^2}\right) = \Omega(k/\log^6 k)$.
\section{Proofs Omitted from Section~\ref{sec: proof of first main theorem}}
\label{--------------------------------------sec proof from alg 1-------------------------}
\label{sec: proofs from alg 1}

\subsection{Proof of Theorem~\ref{thm: cut into expanders}}
Given the graph $H$, we build a new graph $H'$, as follows: Subdivide every edge $e\in E(H)$ with a new vertex $v_e$; add a new vertex $t_e$ and connect it to $v_e$ with an edge. The set of vertices of this new graph $H'$ can be partitioned into three subsets: $V_1=V(H)$; $V_2=\set{v_e\mid e\in E(H)}$ and $\tset=\set{t_e\mid e\in E(H)}$. Let $S=V_1\cup V_2$. We perform a well-linked decomposition of $S$ in graph $H'$, by applying Theorem~\ref{thm: well-linked-decomposition} to it, with parameter $\alpha=\frac{1}{160\log m\cdot\alphaSC(m)}$. Let $\wset$ be the resulting well-linked decomposition of $S$. We define a partition $\wset'$ of the vertices of $H$ as follows: for each $W\in \wset$, we add $W'=W\cap V_1$ to $\wset'$. Our final partition of $V(H)$ is $\wset'$. 

In order to bound $\sum_{W'\in \wset'}|\out(W')|$, observe that each edge $e\in \out(W')$ contributes at least $1$ to $\out(W)$. In addition, $\out(W)$ contains edges connecting some vertices in $V_2$ to the vertices in $\tset$. Such edges do not belong to $H$ and do not contribute to $|\out(W')|$. The total number of such edges in graph $H'$ is $m$. Therefore, $\sum_{W'\in \wset'}|\out(W')|\leq \sum_{W\in \wset}|\out(W)|-m\leq m(1+16\alpha\alphaSC(m)\log m)-m=m/10$.

Finally, we claim that for each $W'\in \wset'$, $\Psi(H[W'])\geq \alpha$. Consider any partition $(A,B)$ of $W'$, and assume without loss of generality that $|E_H(A)|\leq |E_H(B)|$. It is enough to prove that $|E_H(A,B)|\geq \alpha|E_H(A)|$.
We build a partition $(A',B')$ of $W$, as follows. Set $A'$ contains all vertices of $V_1\cap A$, and all vertices $v_e$ where both endpoints of $e$ belong to $A$. (We assume that if both endpoints of an edge $e$ belong to $W$, then so does the vertex $v_e$: otherwise, vertex $v_e$ must belong to a separate cluster $W_e=\set{v_e}$ in $\wset$, and by merging $W$ and $W_e$ we obtain a valid partition.) All other vertices of $W$ belong to $B$. Clearly, $|E_H(A,B)|=|E_{H'}(A',B')|$. Moreover, for every edge $e\in E_{H}(A)$, we have $v_e\in A'$ and $t_e\in \out_{H'}(A')\cap \out_{H'}(W)$, and for every edge $e\in E_{H}(B)$, we have $v_e\in B'$ and $t_e\in \out_{H'}(B')\cap \out_{H'}(W)$. In particular, $|\out_{H'}(A')\cap \out_{H'}(W)|,|\out_{H'}(B')\cap \out_{H'}(W)|\geq |E_H(A)|$. Therefore, from the $\alpha$-well-linkedness of $W'$, $|E_H(A,B)|=|E_{H'}(A',B')|\geq \alpha \min\set{|\out_{H'}(A')\cap \out_{H'}(W)|,|\out_{H'}(A')\cap \out_{H'}(W)|}\geq \alpha|E_H(A)|$.

\subsection{Proof of Theorem~\ref{thm: remove expander}}

\begin{definition} We say that a graph $G=(V,E)$ is an $\alpha$-expander, iff
$\min_{\stackrel{X\subseteq V:}{|X|\leq |V|/2}}\set{\frac{|E(X,\overline{X})|}{|X|}}\geq \alpha$.
\end{definition}

We will use the result of Leighton and Rao~\cite{LR}, who show that
any multicommoditly flow instance in an expander graph that can be
routed with no congestion, can also be routed on relatively short paths with a small
edge-congestion.  In order to use their result, we need to turn $H'$ into a
constant-degree expander\footnote{An alternative to using
  constant degree expanders is to argue about short paths by appealing
  to large conductance and product multicommodity flows.}. We do so as
follows.

We process the vertices of $H'$ one-by-one. Let $v$ be any such
vertex, let $d$ be its degree, and let $e_1,\ldots,e_d$ be the edges
adjacent to $v$. We replace $v$ with a degree-3 expander $X_v$ on $d$
vertices, whose expansion parameter is some constant $\alpha'<1$. Each
edge $e_1,\ldots,e_d$ now connects to a distinct vertex of $X_v$. Let
$H''$ denote the graph obtained after each super-node of $H'$ has been
processed.  Notice that the maximum vertex degree in $H''$ is bounded
by $4$. We next show that graph $H''$ is an $\alpha_0$-expander, for
$\alpha_0=\alpha'\cdot \Psi(H')/12$.

\begin{claim}
  Graph $H''$ is an $\alpha_0$-expander, for $\alpha_0=\alpha'\cdot
  \Psi(H')/12$.
\end{claim}

\begin{proof}
  Assume otherwise, and let $(A,B)$ be a violating cut, that is,
  $|E_{H''}(A,B)|<\alpha_0\cdot\min\set{ |A|,|B|}$.  We use the cut
  $(A,B)$ to define a partition $(A',B')$, of $V(H')$, and show that
  $|E_{H'}(A',B')|< \Psi(H') \cdot
  \min\set{|E_{H'}(A')|,|E_{H'}(B')|}$, contradicting the definition
  of conductance.

  Partition $(A',B')$ is defined as follows. For each vertex $v\in
  V(H')$, if at least half the vertices of $X_{v}$ belong to $A$, then
  we add $v$ to $A'$; otherwise we add $v$ to $B'$.

  We claim that $|E_{H'}(A',B')|\leq |E_{H''}(A,B)|/\alpha'$. Indeed,
  consider any vertex $v\in V(H')$, and consider the partition
  $(A_{v},B_{v})$ of the vertices of $X_{v}$ defined by the partition
  $(A,B)$: that is, $A_{v}=A\cap V(X_{v})$, $B_{v}=B\cap
  V(X_{v})$. Assume without loss of generality that $|A_{v}|\leq |B_{v}|$. Then the
  contribution of the edges of $X_{v}$ to $E_{H''}(A,B)$ is at least
  $\alpha'\cdot |A_{v}|$. After vertex $v$ is processed, we add at
  most $|A_{v}|$ edges to the cut. Therefore,

\[|E_{H'}(A',B')|\leq \frac{|E_{H''}(A,B)|}{\alpha'}\leq \frac{\alpha_0}{\alpha'}\cdot\min\set{|A|,|B|}=\frac{\Psi(H')}{12}\min\set{|A|,|B|}\]

Assume without loss of generality that $\sum_{v\in A'}d_{H'}(v)\leq \sum_{v\in
  B'}d_{H'}(v)$, so $|E_{H'}(A')|\leq |E_{H'}(B')|$.  Consider the set
$A$ of vertices of $H''$, and let $A_1\sse A$ be the subset of
vertices, that belong to expanders $X_{v}$, where $|V(X_{v})\cap
A|\leq |V(X_{v})\cap B|$. Notice that from the expansion properties of
graphs $X_{v}$, $|E_{H''}(A,B)|\geq \alpha'|A_1|$, and so $|A_1|\leq
\frac{|E_{H''}(A,B)|}{\alpha'}\leq \frac{\alpha_0}{\alpha'}|A|\leq
\frac{|A|} 8$.  As every vertex in $A\setminus A_1$ contributes at
least $1$ to the final summation $\sum_{v\in A'}d_{H'}(v)$, we get
that $\sum_{v\in A'}d_{H'}(v)\geq \frac 7 8|A|$, while, as observed
above, $|E_{H'}(A',B')|\leq \frac{\Psi(H')}{12}|A|\leq 0.1|A|$.
Therefore, $|E_{H'}(A')|=\left (\sum_{v\in
    A'}d_{H'}(v)-|E_{H'}(A',B')|\right )/2\geq 0.2|A|$. We conclude
that

\[|E_{H'}(A',B')|\leq  \frac{\Psi(H')}{12} |A|<0.2\Psi(H')|A|\leq \Psi(H')|E_{H'}(A')|,\]

contradicting the definition of conductance.
\end{proof}

The following theorem easily follows from the results of Leighton and
Rao~\cite{LR}, and its proof can be found in~\cite{Chuzhoy11} (see
also \cite{KolmanS06}).

\begin{theorem}\label{thm: routing on expanders}
  Let $G$ be any $n$-vertex $\alpha$-expander with maximum vertex
  degree $\dmax$, and let $M$ be any partial matching over the
  vertices of $G$. Then there is an efficient randomized algorithm
  that finds, for every pair $(u,v)\in M$, a path $P_{u,v}$ of length
  $O(\dmax\log n/\alpha)$ connecting $u$ to $v$, such that the set
  $\pset=\set{P_{u,v}\mid (u,v)\in M}$ of paths causes edge-congestion
  $O(\log^3 n/\alpha)$ in $G$. The algorithm succeeds with high
  probability.
\end{theorem}

Let $X'$ be any degree-$3$ $\alpha'$-expander over
$r''=\Theta(r\Delta^2\log^8k)$ vertices, where $\alpha'$ is some constant. Our
next step is to embed $X'$ into $H''$, by using short
paths. Specifically, we select any collection
$\Gamma'=\set{v_1,\ldots,v_{r''}}$ of vertices in $H''$, and define an
arbitrary $1:1$ matching between the vertices of $X'$ and the vertices
of $\Gamma'$ (we identify these vertices and refer to vertices of $X'$
as $\Gamma'$ from now on). Notice that $r''<r'\leq |E(H')|$, so there are at least $r''$ vertices in $H''$.

Next, for every edge $e=(u,v)\in E(X')$, we find a path $P_e$
connecting $u$ to $v$ in $H''$. This path will serve as the embedding
of the edge $e$. In order to find the embeddings of the edges of $X'$,
we partition $E(X')$ into $5$ disjoint matchings $M_1,\ldots,M_5$,
using the fact that the maximum vertex degree in $X'$ is bounded by
$3$. We then use Theorem~\ref{thm: routing on expanders} to route the
matchings $M_1,\ldots, M_5$ in graph $H''$. Let $\pset=\set{P_e\mid
  e\in E(X')}$ be the resulting set of paths. Then, from
Theorem~\ref{thm: routing on expanders}, the length of every path in
$\pset$ is bounded by $\ell=O(\log^3n)=O(\log^3k)$, and these paths
cause congestion at most $\eta=O(\log^3n)/\Psi(H')=O(\log^5n)=O(\log^5k)$ in $H''$.

We are now ready to define the set $S$ of vertices in graph $H'$. We
add to $S$ every vertex $v$, such that at least one vertex of $X_v$
participates in the paths in $\pset$. Since $|\pset|=r''$, and every
path in $\pset$ contains at most $\ell$ vertices, $|S|\leq r''\cdot
\ell\leq O(r\Delta^2\log^{11}k)\leq r'$, as required. Finally, consider the
graph $G'$, obtained from $H'[S]$, by un-contracting the super-nodes
in $S$. It now only remains to prove that $\tw(G')\geq r$. In order to
do so, we define a subset $\Gamma$ of at least $r''/\Delta$ vertices
of $G'$, and prove that these vertices are $\alpha^{**}$-well-linked in
$G'$, for a suitably large $\alpha^{**}$.

Consider the sub-graph $H^*$ of $H''$, induced by the set $S'$ of
vertices, participating in the paths $\pset$, and the graph $G'$. For
every vertex $v_C\in S$, graph $G'$ contains the sub-graph $G'[C]$,
and graph $H^*$ contains the expander $X_{v_C}$. So we can obtain
$H^*$ from $G'$, by replacing every cluster $G'[C]$ with the expander
$X_{v_C}$. Let $E_0$ be the set of edges in $H^*$, connecting vertices
$(x,y)$ that belong to distinct expanders $X_{v_C},X_{v_{C'}}$. Then for
each edge $e\in E_0$, there is a corresponding edge $e'\in E(G')$,
connecting some vertex $x'\in C$ to some vertex $y\in C'$. We do not
distinguish between the edges $e,e'$, and will think about them as the
same edge.

We now define a subset $\Gamma$ of vertices of $G'$, by mapping every
vertex $x\in \Gamma'$ to its corresponding vertex in $G'$. The mapping
is defined as follows. Consider some vertex $x\in \Gamma'$, and assume
that $x\in X_{v_C}$. Let $e$ be the unique edge of $E_0$ incident on
$x$, and consider the same edge $e$ in graph $G'$. Let $x'$ be the
endpoint of $e$ that belongs to the cluster $C$. We then define
$f(x)=x'$. Let $\Gamma=\set{f(x)\mid x\in \Gamma'}$. Since the degree
of every vertex in $G'$ is at most $\Delta$, $|\Gamma|\geq
r''/\Delta$. From Corollary~\ref{cor: from well-linkedness to
  treewidth}, in order to prove that $\tw(G')\geq r$, it is enough to
show that the set $\Gamma$ of vertices is $\alpha^{**}$-well-linked,
for $\alpha^{**}\geq \frac{6\Delta^2 r}{r''}$.

Consider any partition $(A,B)$ of $V(G')$, and denote
$\Gamma_A=\Gamma\cap A, \Gamma_B=\Gamma\cap B$. Assume without loss of
generality that $|\Gamma_A|\leq |\Gamma_B|$. We need to prove that
$|E(A,B)|\geq \alpha^{**}\cdot |\Gamma_A|$.

Let $\Gamma'_A,\Gamma'_B\subseteq \Gamma'$ be subsets of vertices of
$\Gamma'$, corresponding to the partition $(\Gamma_A,\Gamma_B)$ of
$\Gamma$ (recall that for each vertex $x'\in \Gamma'$, there can be up
to $\Delta$ vertices in $\Gamma$ mapped to $x'$. In this case we only
add one of them to $\Gamma_A$ or $\Gamma_B$). Since graph $X'$ is an
$\alpha'$-expander, there are $|\Gamma_A|$ paths connecting the
vertices in $\Gamma_A'$ to the vertices of $\Gamma_B'$ in $X'$, and
they cause a congestion of $1/\alpha'=O(1)$ in $X'$. Let $\pset_1$
denote this set of paths. Using the embedding of $X'$ into $H''$, we
can build a collection $\pset_2$ of $|\Gamma_A|$ paths in graph $H''$,
where every path connects a distinct vertex in $\Gamma_A'$ to a
distinct vertex in $\Gamma_B'$, and the total congestion due to these
paths is $O(\eta)$. Moreover, from the definition of $H^*$, all paths
in $\pset_2$ are contained in $H^*$. We now use the paths in $\pset_2$
to define a flow $F$ connecting the vertices of $\Gamma_A'$ to the
vertices of $\Gamma_B'$ in $G'$. The flow $F$ follows the paths in
$\pset_2$ on the edges that belong to set $E_0$. In order to complete
the description of this flow, we need to show how to route it inside
the clusters $C$ for $v_C\in S$. For each such cluster $C$, the set
$\pset_2$ of paths defines a set $D_C$ of $O(\eta)$-restricted demands
over the edges of $\out(C)$. Since the cluster $C$ is $\alphaWL$-good,
we can route these demands inside $C$ with congestion at most
$O(\eta\log k/\alphaWL)$. Overall, we obtain a flow $F$ of value
$|\Gamma_A|$, connecting the vertices in $\Gamma_A$ to the vertices of
$\Gamma_B$, with congestion $O(\eta\log k/\alphaWL)$. It follows that
$|E(A,B)|\geq \frac{|\Gamma_A|\alphaWL}{\eta\log
  k}=\Omega\left(\frac{|\Gamma_A|}{\log^{7.5}k}\right )$. We conclude that
set $\Gamma$ is $\Omega(1/\log^{7.5}k)$-well-linked in $G'$. From
Corollary~\ref{cor: from well-linkedness to treewidth}, it follows
that $\tw(G')\geq \Omega\left( \frac{|\Gamma|}{\Delta \log^{7.5}k}\right
)=\Omega\left(\frac{r''}{\Delta^2\log^{7.5}k}\right )\geq r$.

\section{Proof of Claim~\ref{claim: bound on random partition}}

Fix some $1\leq j\leq h+1$. 
Let $\event_1(j)$ be the bad event that $\sum_{v\in X_j}d_{H}(v)\geq \frac{16m}{h}$. In order to bound the probability of $\event_1(j)$, we define, for each vertex $v\in V(H)$, a random variable $x_v$, whose value is $\frac{d_{H}(v)}{r'}$ if $v\in X_j$ and $0$ otherwise. Notice that $x_v\in [0,1]$, and the random variables $\set{x_v}_{v\in V(H)}$ are pairwise independent. Let $B=\sum_{v\in V(H)} x_v$. Then the expectation of $B$, $\mu_1=\sum_{v\in V(H)} \frac{d_{H}(v)}{(h+1) r'}=\frac{2m}{(h+1)r'}$. Using the standard Chernoff bound,

\[\prob{\event_1(j)}<\prob{B> 8\mu_1}\leq 2^{-8\mu_1}=2^{-\frac{16m}{(h+1)r'}}<\frac 1 {6h}\]

since $m\geq \alpha^*k/3$ and $k>hr'\log h/\alpha^*$.

Let $\event_2(j)$ be the bad event that $|E_{H}(X_j)|<
\frac{m}{8h^2}$. We next prove that $\prob{\event_2(j)}\leq \frac 1
{k}$. We say that two edges $e,e'\in E(H)$ are \emph{independent} iff
they do not share any endpoints. Our first step is to compute a
partition $U_1,\ldots,U_z$ of the set $E(H)$ of edges, where $z\leq
2r'$, such that for each $1\leq i\leq z$, $|U_i|\geq \frac m{4r'}$,
and all edges in set $U_i$ are mutually independent. In order to
compute such a partition, we construct an auxiliary graph $Z$, whose
vertex set is $\set{v_e\mid e\in E(H)}$, and there is an edge
$(v_e,v_{e'})$ iff $e$ and $e'$ are not independent. Since the maximum
vertex degree in $G'$ is at most $r'$, the maximum vertex degree in
$Z$ is bounded by $2r'-2$. Using the Hajnal-Szemer\'edi
Theorem~\cite{Hajnal-Szemeredi}, we can find a partition
$V_1,\ldots,V_z$ of the vertices of $Z$ into $z\leq 2r'$ subsets,
where each subset $V_i$ is an independent set, and $|V_i|\geq
\frac{|V(Z)|}{z}-1\geq \frac{m}{4r'}$. The partition $V_1,\ldots,V_z$
of the vertices of $Z$ gives the desired partition $U_1,\ldots,U_z$ of
the edges of $H$. For each $1\leq i\leq z$, we say that the bad event
$\event_2^i(j)$ happens iff $|U_i\cap E(X_j)|<
\frac{|U_i|}{2(h+1)^2}$. Notice that if $\event_2(j)$ happens, then
event $\event_2^i(j)$ must happen for some $1\leq i\leq z$. Fix some
$1\leq i\leq z$. The expectation of $|U_i\cap E(X_j)|$ is
$\mu_2=\frac{|U_i|}{(h+1)^2}$. Since all edges in $U_i$ are
independent, we can use a standard Chernoff bound to bound the
probability of $\event_2^i(j)$, as follows:

\[\prob{\event_2^i(j)}=\prob{|U_i\cap E(X_j)|<\mu_2/2}\leq e^{-\mu_2/8}=e^{-\frac{|U_i|}{8(h+1)^2}}\]

Since $|U_i|\geq \frac{m}{4r'}$, $m\geq k\alpha^*/3$, $k\geq 2^{10}h^2r'\log
k/\alpha^*$, this is bounded by $\frac{1}{k^2}$. We conclude that
$\prob{\event_2^i(j)}\leq \frac{1}{k^2}$, and by using the union bound
over all $1\leq i\leq z$, $\prob{\event_2(j)}\leq \frac{1}{k}$.

Using the union bound over all $1\leq j\leq h+1$, with probability at
least $\half$, none of the events $\event_1(j),\event_2(j)$ for $1\leq
j\leq h+1$ happen, and so for each $1\leq j\leq h+1$,
$|\out_{H}(X_j)|\leq\sum_{v\in X_j}d_H(v)< \frac{16m}{h}$, and
$|E_{G'}(X_j)|\geq\frac{m}{8h^2}$ must hold.

\section{Proof of Lemma~\ref{lem:thomassen}}

The proof closely follows the proof of Proposition 2.1
in~\cite{Thomassen88}.  Let $G$ be any graph, and $k$ any integer. If
$\tw(G)\geq h_{\fset}(k)$, then from our assumption, $p_{\fset}(G)\geq
k$, and there is nothing to prove. So from now on, it is enough to
prove the following. If $G$ is any graph with $\tw(G)=w<
h_{\fset}(k)$, then either $p_{\fset}(G)\geq k$, or $c_{\fset}(G)\leq
k(w+1)$. We prove this statement by induction on $k$. The statement is
clearly true for $k=0$. Consider now some general value of $k$.

Let $T$ be the tree-decomposition of width $w$ of $G$. For each vertex
$v\in V(T)$, we denote by $X_v$ the corresponding subset of vertices
of $G$, and recall that $|X_v|\leq w+1$. For each sub-tree $T'\sse T$,
we denote by $S_{T'}=\bigcup_{v\in V(T')}X_v$, and by $G_{T'}$ the
sub-graph of $G$ induced by $S_{T'}$.

For every vertex $v\in V(T)$, we consider all pairs $(T_1,T_2)$ of
sub-trees of $T$, where $T_1\cup T_2=T$, and $T_1\cap
T_2=\set{v}$. Among all such triples $(v,T_1,T_2)$, we are interested
only in those where $G_{T_1}$ contains a sub-graph isomorphic to a
graph in $\fset$, and among all triples satisfying this condition, we
select the one minimizing $|V(T_1)|$. Let $H$ be any sub-graph of
$G_{T_1}$ isomorphic to a member of $\fset$. Then $V(H)\cap X_v\neq
\emptyset$, since otherwise we can obtain a new triple
$(v',T_1',T_2')$ satisfying all the above properties, with
$T_1'\subsetneq T_1$, contradicting the minimality of $T_1$.

Assume now that $c_{\fset}(G)> k(w+1)$. In other words, for any subset
$A$ of $k(w+1)$ vertices in graph $G$, $G\setminus A$ contains a
sub-graph isomorphic to a graph in $\fset$. In particular, if we let
$G'=G\setminus X_v$, then for any subset $A$ of $(k-1)(w+1)$ vertices
in this graph, $G'\setminus A$ contains a sub-graph isomorphic to a
graph in $\fset$. By the induction hypothesis, this means that $G'$
contains $(k-1)$ disjoint sub-graphs $G_1,\ldots,G_{k-1}$, each of
which is isomorphic to a graph in $\fset$. Moreover, each such graph
$G_i$ must be disjoint from $G_{T_1}$ (since, as observed above, any
copy of a graph in $\fset$, which is contained in $G_{T_1}$, must
intersect $X_v$). Let $H$ be any copy of a graph in $\fset$ that is
contained in $G_{T_1}$. Then $G_1,\ldots,G_{k-1},H$ are $k$ disjoint
subgraphs of $G$, each of which is isomorphic to a graph in $\fset$,
as required.

\section{Proof of Lemma~\ref{lem:ep-improved}}

The proof is inspired by the argument in \cite{FominST11}.

  We prove that for each $k\geq 1$, and for each graph $G$, if
  $p_{\fset}(G)\le k$, then $c_{\fset}(G)\leq \beta h_\mF(k) \log (k+1)$.
  The proof is by induction on $k$. The claim is trivially true for
  $k=0$.  We prove the statement for $k \ge 1$ assuming that it holds
  for all values up to $k-1$.
  
  Let $G$ be such that $p_\fset(G) = k$ and let $T=(V_T,E_T)$ be
  a tree decomposition of smallest width for $G$. We observe that
  the width of $T$ is strictly less than $h_\fset(k+1)$ for otherwise
  $p_\fset(G) > k$, contradicting our assumption. For $t \in V_T$
  let $X_t \subseteq V$ be the bag of vertices at $t$. 
  We root $T$ at any vertex and use the following notation. For $t \in V_T$, 
  $T_t$ is the subtree of $T$ rooted at $t$. 
  $G_t= G[S_t]$ where $S_t = \cup_{t' \in T_t} X_{t'}$.
  $G_t^- = G_t \setminus X_t$ is the graph obtained by removing the
  nodes in $X_t$ from $G_t$.

  The induction step is based on the following claim.
  \begin{claim}
    There exists a separator $S \subseteq V$ such that $|S| \le 2 h_\fset(k+1)$
    and for each connected subgraph $G'$ in $G\setminus S$, 
    $p_\fset(G') \le \floor{2k/3}$.
  \end{claim}

\begin{proof}
  Call a node $t \in V_T$ large if $G_t^-$ contains a connected
  subgraph $G'$ such that $p_\fset(G') > \floor{2k/3}$. Otherwise
  $t$ is small.  If the root $r$ is small then $X_r$ is the desired
  separator and we are done. Otherwise, let $t$ be the deepest large
  node in $T$. 
  There is a single connected component $G'$ in $G \setminus X_t$ such
  that $p_\fset(G') > \floor{2k/3}$, otherwise it would imply that
  $p_\fset(G) > k$.  Moreover $G'$ is contained in $G_{t'}$ for some
  child $t'$ of $t$. We claim that $S = X_{t} \cup X_{t'}$ is the
  desired separator.  If $G\setminus S$ still contains a connected
  component $G'$ such that $p_\fset(G') > \floor{2k/3}$ then it is
  contained in $G_{t'}^-$ contradicting the choice that $t$ is the
  deepest large node in $T$.
  \end{proof}
  
  Let $S$ be the separator from the claim. We have $|S| \le 2
  h_\fset(k+1)$ which by the assumption on the function $h(\cdot)$ is
  at most $2\alpha h_\fset(k)$.  Let $G_1,G_2,\ldots, G_\ell$ be the
  connected components of $G\setminus S$ and let $k_i = p_\fset(G_i)$.
  For $1 \le i \le \ell$, $k_i \le \floor{2k/3} < k$, and moreover
  $\sum_{i=1}^\ell k_i \le k$.  Let $S_i$ be a minimum cardinality
  $\fset$-cover for $G_i$. From the induction hypothesis $|S_i| \le
  \beta h_\fset(k_i) \log (k_i+1)$.  Since $\fset$ is a family of
  connected graphs, we note that $S' = S \cup\left ( \bigcup_i
    S_i\right )$ is a $\fset$-cover in $G$ whose cardinality can be
  bounded as $2 \alpha h_\fset(k) + \sum_i \beta h_\fset(k_i) \log
  (k_i+1)$.  If $k=1$ then $k_i = 0$ for all $i$ and therefore $|S'|
  \le 2 \alpha h_\fset(k)$ which proves the induction hypothesis for
  $k=1$ if $\beta \ge 2 \alpha$. We will now assume $k \ge 2$ in which
  case for each $i$, $k_i +1 \le \floor{2k/3}+1 \le 3(k+1)/4$. The
  cardinality of $S'$ is upper bounded as:
  \begin{eqnarray*}
    2 \alpha h_\fset(k) + \sum_i \beta h_\fset(k_i) \log (k_i+1) & \le & 2 \alpha h_\fset(k) + \sum_i \beta  h_\fset(k_i) \log (\frac{3}{4}(k+1)) \\
    & \le & 2 \alpha h_\fset(k) + \beta  \log (\frac{3}{4}(k+1)) \sum_i h_\fset(k_i) \\
    & \le & 2  \alpha h_\fset(k) + \beta  \log (\frac{3}{4}(k+1)) \cdot h_\fset(k)     \quad \quad \mbox{(since $h_\fset(\cdot)$ is superadditive)} \\
    & \le & 2  \alpha h_\fset(k) - \beta  h_{\fset}(k) \log \frac{4}{3}  + \beta h_\fset(k) \log (k+1) \\
    & \le & \beta  h_\fset(k) \log (k+1),
  \end{eqnarray*}
  where the last inequality follows by choosing $\beta$ sufficiently
  large compared to $\alpha$. This establishes the induction step for
  $k$ and finishes the proof.

\fi

\end{document}